\documentclass{article}[11pt]
\usepackage{float}
\usepackage{graphicx} %
\usepackage{amsmath, amssymb, amsthm}
\usepackage{thm-restate}
\usepackage{algorithm}
\usepackage[noend]{algpseudocode}
\usepackage{tikz}
\usetikzlibrary{decorations.pathreplacing,calc}
\usepackage{appendix}
\usepackage{fullpage}
\usepackage{enumitem}
\usepackage{cleveref}
\usepackage{epigraph}

\newtheorem{lemma}{Lemma}
\newtheorem{theorem}{Theorem}
\newtheorem{definition}{Definition}

\newcommand{\E}{\mathbb{E}}

\usepackage{todonotes}
\definecolor{xred}{HTML}{F19F9E}
\definecolor{xblue}{HTML}{9ECFF1}

\usepackage[most]{tcolorbox} %

\tcbset{
  doubleframe/.style={
    enhanced,
    breakable,
    colback=white,        %
    colframe=black,       %
    boxrule=0.6pt,        %
    borderline={0.6pt}{2pt}{black}, %
    sharp corners,
    left=10pt,right=10pt,top=8pt,bottom=8pt,
    before skip=8pt, after skip=8pt
  }
}

\newtcolorbox{boxedremark}[1][]{doubleframe,title={#1}}

\newtcbtheorem[number within=section]{rulebox}{Rule}%
{doubleframe,fonttitle=\bfseries}{rule}

\title{Maintaining Random Assignments under Adversarial Dynamics}
\author{Bernhard Haeupler\thanks{Partially funded by
the Ministry of Education and Science of Bulgaria’s support for INSAIT as part of the Bulgarian National Roadmap
for Research Infrastructure and through the European Research Council (ERC) under the European Union’s Horizon
2020 research and innovation program (ERC grant agreement 949272).}\\
\normalsize INSAIT, Sofia University “St. Kliment Ohridski” \& ETH Zurich
\and
Anton Paramonov\\
\normalsize ETH Zurich\\
}

\date{}

\begin{document}
\maketitle

\begin{abstract}
We study and further develop powerful general-purpose schemes to maintain random assignments under adversarial dynamic changes.  
\smallskip

The goal is to maintain assignments that are (approximately) distributed similarly as a completely fresh resampling of all assignments after each change, while doing only a few resamples per change. This becomes particularly interesting and challenging when dynamics are controlled by an adaptive adversary.%
\smallskip

Our work builds on and further develops the proactive resampling technique [Bhattacharya, Saranurak, and Sukprasert ESA'22]. We identify a new ``temporal selection'' attack that adaptive adversaries can use to cause biases, even against proactive resampling. We propose a new ''temporal aggregation'' principle that algorithms should follow to counteract these biases, and present two powerful new resampling schemes based on this principle. 

\smallskip

We give various applications of our new methods. The main one in maintaining proper coloring of the graph under adaptive adversarial modifications: we maintain $O(\Delta)$ coloring for general graphs with maximum degree $\Delta$ and $O(\frac{\Delta}{\ln \Delta})$ coloring for triangle free graphs, both with sublinear in the number of vertices average work per modification. Other applications include efficiently maintaining random walks in dynamically changing graphs. 
\end{abstract}

\section{Introduction}

Randomness is a remarkably effective organizing principle across many computational and algorithmic settings. It naturally balances load, prevents concentration of adversarial influence, and ensures that local irregularities do not aggregate into global bias. Yet, in many real-world or algorithmic systems, the underlying structure changes continuously—participants join or leave, resources become unavailable, or edges in a network appear and disappear. The challenge, then, is not merely to generate a good random assignment once, but to \emph{maintain} such assignments as the system evolves.  

The central question of this work is:
\begin{quote}
\emph{How can we maintain random-like assignments under ongoing, potentially adversarial dynamics, while making only a few updates per change?}
\end{quote}

We begin by illustrating why random assignments are beneficial through three settings. These settings are of increasing generality and complexity and will serve as running examples throughout this paper.

\textbf{Participant–Group Assignment.}
Consider a setting with many participants that must be divided into several groups. A good example is a distributed peer-to-peer system in which groups are used to distribute work or storage (of a distributed hash table) among participants. A fraction of these participants may be malicious, yet their identities are not known. In this setting, assigning each participant to a group uniformly at random offers several advantages. First, it ensures that group sizes remain balanced. Second, if the overall fraction of malicious participants is less than one half, then with high probability, each group will contain a majority of honest participants. Finally, random assignments satisfy many natural fairness notions (e.g., no participant can choose its group) and minimize the chances of small groups of colluding participants being placed in the same group.

\textbf{Job–Machine Assignment.}
Consider a system consisting of jobs and machines, where each job can be processed by several routines eligible for this job, and each routine is performed by a subset of machines. 
A natural randomized strategy to achieve effective load balancing without requiring centralized coordination is to assign each job to a uniformly random routine. 

\textbf{Node–Random Walk Assignment.}
Random walks and branching random walks have a wide variety of applications. The famous PageRank \cite{bahmani2010fast} measure of importance corresponds to expected times to be visited when running a random walk of geometric length from every node in a network. Random walks are also used in popular ML features for node similarity in networks, including node2vec \cite{grover2016node2vec}, using the idea that branching random walks started from similar nodes encounter the same nodes or regions of a graph more often than dissimilar ones. Overall, assigning to each node in a network one or multiple (branching) random walks and using the dispersion, concentration, mixing, and other properties of random walks to estimate, define, and measure properties of nodes within a network has been a powerful and widely used paradigm across sciences. Random walks can also be used to distribute traffic evenly and keep network congestion low, as for example done by many oblivious routing schemes~\cite{broder1992existence, broder1997static, akyildiz2000new, tian2005randomwalk, spyropoulos2005spray, servetto2002constrained}.

In each of these settings, pure randomization performs well in a static world. But when the system evolves dynamically, e.g., participants join and leave, or edges are being deleted, naive strategies fail and assignments lose their random properties.

A natural countermeasure against this is to rerandomize the entire system after each modification, thereby restoring a perfectly random configuration. Yet this comes with an enormous computational overhead—one that might not even be necessary: perhaps similar distributions can be achieved with only limited, carefully chosen updates.

The difficulty of maintaining a (close to) random assignment increases even further when system updates are not arbitrary but chosen adversarially in response to the algorithm's previous actions. An adaptive adversary might, for instance, repeatedly join and leave the system until it lands in a particular group, or selectively delete edges to steer random walks toward a specific region of the graph. These manipulations can create highly biased, undesirable distributions.

This paper focuses on precisely this problem: how to preserve strong distribution guarantees efficiently even under adaptive, potentially malicious dynamics. We identify limitations and vulnerabilities of existing methods and introduce new algorithmic tools that maintain provable distribution guarantees with low resampling cost.

\section{Related Work}

We now formally define the two settings - the Participant-Group Assignment and Job-Machine Assignment. We provide especially in-depth formulations so as to present concrete, formal examples of the results that can be obtained. 

\subsection{Participant-Group Setting: Problem-specific Cuckoo/Rotate Rules}

We start by looking more formally at the participant-group setting in which a large number of participants must be continuously assigned to $g$ groups. Participants are indistinguishable to the system: when they join, leave, or rejoin, the assignment algorithm cannot tell whether it is interacting with a new or returning individual. The system evolves dynamically as participants repeatedly join and leave, and upon each join, the newcomer is assigned uniformly at random to one of the groups. If the schedule of joins and departures is \emph{oblivious}, that is, independent of the random outcomes of the assignment process, then the configuration at any fixed time is distributed uniformly and satisfies all previously mentioned desirable fairness, robustness, and load-balancing properties. However, if some participants are malicious, they can leave and immediately rejoin the system, thereby triggering an \emph{adversarial resample}, and these resamples can be scheduled \emph{adaptively}: a participant may decide whether to trigger a resample after observing its current assignment or state of the system. This adaptivity enables malicious participants to severely bias the resulting assignment in what is commonly called a \emph{join-leave attack}. 

Theoretically sound methods for preventing such attacks were proposed by Scheideler and Awerbuch in \cite{awerbuch2006towards,scheideler2005spread}. We paraphrase their ideas and results for the context of the group assignment setting, which we define formally next:

\begin{boxedremark}[Participant-Group Assignment Setting]
There are $n$ participants, $g$ groups, and a $\beta$ fraction of malicious participants. Initially, every participant gets assigned a random group. At every time step, the adversary can adaptively select one malicious participant for an \emph{adversarial resample}. 
\end{boxedremark}

It is easy to see that in this setting an adaptive adversary can force malicious participants to become a majority in some fixed group at some fixed time, even for $\beta$ as small as $1/g$. For this, every malicious participant simply forces an adversarial resample through a leave-rejoin until they are assigned to a chosen fixed group. We generally call this type of attack a \emph{resample-until-successful} attack.   

In \cite{awerbuch2006towards} and \cite{scheideler2005spread}, two rules were proposed to circumvent such attacks. The {\bfseries $k$-cuckoo rule} assigns a joining participant a random group and then uniformly randomly selects $k$ participants assigned to this group and reassigns each of these $k$ participants to a new random group. The {\bfseries $k$-rotation rule} takes a newly joined participant, selects a random existing participant, and swaps them. The replaced existing participant is then again swapped with another existing participant. This is done $k-1$ times, and finally, the participant replaced last is assigned to a random group. 

The overhead for both rules is $k$, i.e., for any adversarially resampling, only $k$ extra resamples are performed. Still, as long as the fraction of malicious participants is less than $50\%$
a constant $k$ suffices to ensure that the probability that the adversary can achieve a majority of malicious participants in some group is exponentially small in the average group size. In particular, \cite{awerbuch2006towards,scheideler2005spread} proves that choosing $k = \Theta(\frac{1}{0.5 - \beta})$ suffices for both rules. (see also Appendix \ref{sec:appendix:cuckoo})

Achieving almost the same exponential Chernoff-type concentration bounds for the probability of malicious participants becoming a majority in some group while only doing constantly many resamples per join is a powerful demonstration of what careful injection of fresh randomness into an adversarially controlled dynamic system can achieve. 

From the point of view of this paper, the above rules and results have the drawback that they are quite \emph{problem-specific}. It is not clear at all how to transfer and apply either of these rules in other contexts, such as maintaining random walks in a dynamically changing graph. Furthermore, both rules are also problem-specific in the group-assignment setting itself in that they do not try to generate a random assignment that can not be biased by the adversary, but instead specifically aim to avoid majorities of malicious nodes. In particular, adversarial dynamics can easily still generate many severely non-random biases. For example, via the same \emph{resample-until-successful} attack a malicious participant can easily choose its own group and more generally an adaptive adversary can easily ensure that some specific $b$ malicious participants rendezvous at a specific time at a specific group as long as $b$ is smaller than a $\frac{1}{k}$ fraction of the average group size. Many fairness and collusion properties, other than the majority safety it is designed to ensure, are not even approximately maintained by the rotation or cuckoo rules, despite holding for truly random assignments with exponentially strong probabilities.

\subsection{Proactive Resampling and the Job-Machine Setting}
The beautifully simple \emph{proactive resampling rule} was first introduced in the paper by Bernstein et al. \cite{bernstein2022fully} and further developed by Bhattacharya et al. \cite{bhattacharya2022simple}.

\begin{boxedremark}[Proactive Resampling Rule]
If an object is (adversarially) sampled at time $t$, also proactively resample it at times $t+2^i$ for all $i\geq0$.
\end{boxedremark}

Generic as it is, this rule can be readily applied in a variety of contexts. For instance, when the participant joins at time $t$, assign them to a uniformly random group, and then reassign them to a new group at times $t + 2^i$. Or if the adversary deletes an edge through which a random walk was passing at time $t$, resample this random walk at time $t$, and then further at times $t + 2^i$.

Furthermore, this rule is lightweight: for a single adversarial sample, this mechanism triggers only logarithmically many additional resamples. 

But most importantly, even without any knowledge of the underlying problem structure, proactive resampling produces assignments with very strong (pseudo-)random properties. 
For example, when applied to the participant-group setting, it ensures that the probability that a fixed participant is assigned to a fixed group at a fixed time is at most $\frac{O(\log n)}{g}$. Moreover, the probability that a given subset of participants is assigned to the same group at a given time is exponentially small in the size of the subset. In general, expectations change by merely polylogarithmic factors, and strong Chernoff-type concentration bounds hold. 

The work \cite{bhattacharya2022simple} further proposed an application-independent framework to reason about random assignments in dynamic systems.

\begin{boxedremark}[Job-Machine Setting]
    Consider a bipartite hypergraph where the left part nodes represent machines and the right part nodes represent jobs. Every hyperedge, also called a routine, in this graph contains exactly one job. A routine of the form $\{j, m_1, \ldots, m_k\}$ represents the fact that job $j$ can be assigned to a group of machines $\{m_1, \ldots, m_k\}$. A valid assignment is a subset of routines such that every job belongs to exactly one routine in this subset. Throughout the process, an algorithm must maintain a valid assignment. If a routine containing job $j$ and machine $m$ is in the algorithm's assignment, we say that $j$ uses $m$.

    At each point in time, an adversary can delete a machine. If a job is using a deleted machine, an algorithm must reassign it to an incident routine that is still present in order to maintain a valid assignment. 

    The load of a machine at a given time is the number of jobs using it at that time. 
\end{boxedremark}

When resampling a job, the proactive resampling algorithm in \cite{bhattacharya2022simple} simply picks a uniformly random routine incident to that job that still exists. The analysis then shows the following strong \mbox{(pseudo-)randomness} properties for the loads of machines: 

\begin{restatable}{claim}{incorrectJobMachineLoad}\textnormal{[Lemma 14~\cite{bhattacharya2022simple}]}
\label{clm:incorrectJobMaschieneLoad}
Consider the job-machine setting with $n$ jobs in which machines are deleted by an adaptive adversary. Under proactive resampling, with high probability, the number of jobs using a machine $x$ at time $t$ is at most $O(\log n) \cdot L + O(\log n)$, where $L$ is the expected number of jobs using $x$ by a fresh uniformly random assignment at time $t$.
\end{restatable}

This claim is used in \cite{bhattacharya2022simple} to maintain a spanner in dynamic graphs with low recourse, i.e., the number of edge changes in the maintained spanner per graph update.

Overall, the work~\cite{bhattacharya2022simple} identified proactive resampling as a general approach for maintaining randomized assignments under dynamic changes. However, their job-machine framework comes with two major limitations. First, it assumes that routines are sampled uniformly, which restricts its applicability to non-uniform distributions such as the distribution of random walks. Second, the analysis focuses on a specific property of the resulting random assignment—namely, the machine loads—rather than addressing the more general question of which distributions an adversary can generate in a dynamic setting.

\subsection{Balls and Bins}
The recent work by Fine, Kaplan, and Stemmer \cite{fine2025minimizing} revisits the problem of maintaining a spanner in dynamic graphs considered in \cite{bhattacharya2022simple}. They show that a much simpler setting, compared to job and machines, suffices for the problem.

\begin{boxedremark}[Balls and Bins Setting]
    Initially, the system consists of $n$ balls distributed among $n$ bins. At any point in time, an adversary can pick a bin to delete. All the balls contained in the deleted bin must be redistributed. This process continues until there is a single bin (containing all the balls) left.

    The number of balls redistributed at a given time step is called a recourse of that step, and the sum of the recourses of all steps is called a total recourse.
\end{boxedremark}

An elegant observation of \cite{fine2025minimizing} is that redistributing the balls from a deleted bin by independently placing each ball uniformly at random into one of the remaining bins yields an asymptotically optimal total recourse of $O(n \log n)$, even against an adaptive adversary (i.e., one that observes the contents of all bins after each step and chooses which bin to delete accordingly). Moreover, \cite{fine2025minimizing} leverages this strategy to maintain $3$-spanners in dynamic graphs, improving the recourse bounds compared to \cite{bhattacharya2022simple}.

Further, \cite{fine2025minimizing} proposes a generalization of a Job-Machine framework where, for a given job, different routines incident to this job have arbitrary probabilities of being sampled. In particular, in the beginning, for each job $j$ and routine $r$ incident to $j$, the probability $p_j(r)$ is specified. Then, when job $j$ is resampled at time $t$, it is assigned to routine $r$ with probability $\frac{p_j(r)}{\sum_{r' \in E_t(j)}p_j(r')}$, where $E_t(j)$ are routines incident to $j$ that are still present at time $t$. Authors of \cite{fine2025minimizing} then provide an analysis of their simple strategy in this extended setting. 

Although more general than uniform Job-Machines, this setting may still not capture the behavior of some dynamic systems, e.g., dynamic graphs with random walks, where not only different routines can have different probabilities of being sampled, but also these probabilities can be chosen by an adaptive adversary throughout the execution.

\section{Our Results}
In this paper, we address current gaps and substantially broaden the scope of random assignments under adversarial dynamics. Our contributions can be summarized as follows.

{\bfseries 1. Generalized Framework:} 
The job–machine framework of \cite{bhattacharya2022simple} considers uniformly random assignments of jobs to eligible routines. In \cite{fine2025minimizing}, assignment distribution may not be uniform, but is predefined from the start. However, in many settings, non-uniform dynamic distributions arise naturally and are in fact necessary. For example, consider assigning a node $v$ to a random walk of length $l$ starting from $v$. It would be incorrect to select a uniformly random path of length $l$ starting at $v$, since a random walk $(v, u_1, u_2, \ldots, u_l)$ should be sampled with probability $\left(\deg(v)\cdot\prod_{i < l} \deg(u_i)\right)^{-1}$ rather than uniformly. Moreover, the degrees of the nodes change in dynamic graphs, thus changing the probability of sampling a given random walk. To model such random-walk assignments and other natural dynamic non-uniform processes, a more general view is required. 

We propose a framework that captures this full generality: it allows for arbitrary distributions that may also evolve arbitrarily over time. This framework, formally described in Section~\ref{sec:framework}, is designed to facilitate comparisons between static and dynamic systems. In the static setting, we refer to \emph{static distributions}, meaning the distributions obtained when all objects are sampled simultaneously, and to \emph{static realizations}, i.e., particular outcomes drawn from these distributions. In contrast, in the dynamic setting, we consider \emph{joint distributions}, where different objects may be sampled at different times, and correspondingly \emph{joint realizations}, which represent the resulting composite outcomes.

{\bfseries 2. Temporal Selection Attack:} 
One advantage of the job–machine framework introduced in~\cite{bhattacharya2022simple} over the earlier work of~\cite{bernstein2022fully}, which first proposed proactive resampling, is its ability to handle different (uniform) sampling probabilities for different jobs and to allow these probabilities to evolve over time as machines are deleted. This extension makes results such as \Cref{clm:incorrectJobMaschieneLoad} both more expressive and more subtle to analyze. Unfortunately, proactive resampling does not provide the guarantees claimed in \Cref{clm:incorrectJobMaschieneLoad}, even within a restricted job–machine setting.

We demonstrate this through the \emph{Temporal Selection Attack}—an adversarial strategy that refutes \Cref{clm:incorrectJobMaschieneLoad}. Unlike the resample-until-successful attack, where the adversarial power lies in the ability to perform many samples of each object, in temporal selection, the adversary leverages the ability to modify distributions of the objects. In particular, we observe that it suffices to sample an object only once, but at the time when this object has a high probability of a ``bad'' outcome. In its essence, the temporal selection does the following. 

\begin{boxedremark}[Temporal Selection Attack]
    \textbf{for} $u$ in objects \textbf{do}\\
    \hspace*{0.6cm} Change distributions so that $u$ is likely to have a bad outcome\\
    \hspace*{0.6cm} Sample $u$
\end{boxedremark}
This attack goes completely unnoticed for the proactive resampling, which was designed to counter the resample-until-successful approach. In particular, we show the following theorem that disproves \Cref{clm:incorrectJobMaschieneLoad}.

\begin{restatable}{theorem}{temporalSelectionJobMachines}
    There exists a configuration of $n$ jobs, machines, and routines, along with an adaptive adversarial sequence of machine deletions of $T = O(n^3)$ steps, such that under proactive resampling
    \begin{itemize}[noitemsep, topsep=-1pt]
        \item There exists a machine $S$ with $\Omega(\sqrt{n})$ jobs using $S$ at time $T$ in expectation, and
        \item For any time $t \leq T$, the expected number of jobs using $S$ if all jobs are freshly sampled at time $t$ is $O(1)$.
    \end{itemize}
    \vspace{-0.5em}
\end{restatable}

Moreover, we show how this attack can be instantiated in the context of random walks on dynamic graphs to overload specific edges or vertices. We believe that understanding such attacks is crucial for the safe design of random-walk–based algorithms such as PageRank and node2vec.

{\bfseries 3. Algorithms:}
Our algorithm design is driven by two primary goals: (i) to perform only a small number of additional resamples per adversarial resampling, and (ii) to maintain a joint distribution that remains similar to the static one. The latter we achieve through what we call a \emph{Temporal Aggregation Principle}, which is inspired by the insights from the temporal selection attack. In particular, observe that in the extreme case of the attack, every object is sampled at its own time step. The temporal aggregation principle combats that by striving to sample objects simultaneously.
\begin{boxedremark}[Temporal Aggregation Principle]
    The algorithm must ensure that at any point in time, realizations of objects come from a few different time steps.
\end{boxedremark}

Under this principle, the joint distribution is \emph{spliced} from a few static ones, which naturally leads to it resembling static properties. While we study this resemblance in general through table games, we also focus on an important special case—when the joint distribution preserves \emph{loads}. We formally define loads in Section~\ref{sec:algorithms}, but the concept is intuitive: the load of a joint realization may represent “how many malicious participants are assigned to a given group,” “how many jobs are using a given machine,” or “how many random walks pass through a given edge.” 

We establish strong guarantees on load preservation for our proposed algorithms. The first of these, \texttt{Greedy Temporal Aggregation}, is a naive application of the temporal aggregation principle. It is remarkably simple and enjoys the following properties:

\begin{restatable}{theorem}{GreedyTemporalAggregationThm}
    \label{thm:greedy temporal aggregation main}
    If the system evolves for $T$ time steps, and in each time step a static realization has a load at most $L$ with probability at least $1 - \varepsilon$, then Greedy Temporal Aggregation ensures that at time $T$ the joint realization has a load at most $L\cdot O(\log n)$ with probability at least $1 - \varepsilon \cdot T$.

    Moreover, if the adversary does $q$ samples by the time $T$, Greedy Temporal Aggregation performs no more than $q\cdot O(\log n)$ samples.
\end{restatable}

Our second algorithm \texttt{Landmark Resampling} adapts proactive resampling to follow the temporal aggregation principle.

\begin{boxedremark}[Landmark Resampling]
    If an adversary samples an object at time $t_0$, Landmark Resampling schedules its resamples at times $t_1$, $t_2$, $t_3$, $\ldots$ with $t_i$ defined as $t_i = t_{i-1} + 2^i + x$, where $x \in [0,2^{i}]$ is chosen to maximize the number of trailing zeros in the binary description of $t_i$.
\end{boxedremark}

While greedy temporal aggregation already provides meaningful guarantees, landmark resampling achieves strictly stronger results. The algorithm  satisfies the following theorem:

\begin{restatable}{theorem}{LandmarkResamplingThm}
    \label{thm:landmark resampling}
     If the system evolves for $T$ time steps, and in each time step a static realization has a load at most $L$ with probability at least $1 - \varepsilon$, then Landmark Resampling ensures that at time $T$ the joint realization has a load at most $L\cdot O(\log T)$ with probability at least $1 - \varepsilon \cdot O(\log T)$.

    Moreover, if the adversary does $q$ samples by the time $T$, Landmark Resampling performs no more than $q\cdot O(\log T)$ samples.
\end{restatable}

{\bfseries 4. Table Games:}
Although many applications can be reduced to bounding loads, this work instead aims to characterize the full range of distributions an adversary can produce. To make these distributions tractable, we introduce table games — simple, transparent processes that mimic the algorithm–adversary interaction in our framework and yield distributions that are (almost) identical but far easier to analyze. We present table games for proactive resampling, greedy temporal aggregation, and landmark resampling, which simplifies comparison of their guarantees and, in the case of proactive resampling, exposes how an adversary can craft an attack that produces a highly skewed joint distribution.

{\bfseries 5. Applications:}
We demonstrate the applicability of our framework and algorithms across a range of problems.

We show that greedy temporal aggregation can be used to patch proof of the main result of~\cite{bhattacharya2022simple}. 

We further apply our techniques to dynamic graph coloring. In this setting, the algorithm must maintain a proper coloring of the graph that undergoes adaptive adversarial edge deletions and insertions. The goal is to minimize both the number of colors used at each time step and the computation performed per adversarial modification. 

We obtain the following result for general maximum degree $\Delta$ graphs, ``dynamizing'' the $\widetilde{O}(n\sqrt{n})$ coloring algorithm by Assadi, Chen and Khanna \cite{assadi2019sublinear}.

\begin{restatable}{theorem}{dynamicSublinearColoringThm}
    \label{thm:dynamic sublinear coloring}
    Consider a sequence of graphs on $n$ vertices $G_1, \ldots, G_T$ where $G_{t+1}$ is obtained from $G_t$ through adaptive adversarial edge deletions and insertions. Suppose that there exists an integer $\Delta$ such that for every $t \in [T]$, the maximum degree in $G_t$ is no more than $\Delta$. For any $\varepsilon > 0$, there exists an algorithm that with high probability maintains a proper $O(\frac{1}{\varepsilon}\cdot \Delta)$ coloring of $G$ and if an adversary performs $q$ edge deletions and insertions, the algorithm performs a total of $\widetilde{O}(q \cdot n^{1/2 + \varepsilon})$ computation. 
\end{restatable}

For triangle-free graphs, we achieve a coloring with substantially fewer, namely, $O(\frac{\Delta}{\ln \Delta})$ colors while still keeping the computation per adversarial modification sublinear in $n$. This result dynamizes the algorithm of Alon and Assadi \cite{beyondPalette}.

\begin{restatable}{theorem}{dynamicSublinearTriangleFreeColoringThm}
    Consider a sequence of graphs on $n$ vertices $G_1, \ldots, G_T$ where $G_{t+1}$ is obtained from $G_t$ through adaptive adversarial edge deletions and insertions. Suppose that for every $t \in [T]$ $G_t$ is guaranteed to be triangle-free; further, there exists an integer $\Delta$ such that the maximum degree in $G_t$ is no more than $\Delta$. For any $\gamma, \varepsilon \in (0, 1)$, there exists an algorithm that with high probability maintains an $O(\frac{1}{\varepsilon}\cdot \frac{\Delta}{\gamma \cdot \ln \Delta})$ proper coloring of $G$ and if an adversary performs $q$ edge deletions and insertions, the algorithm performs a total of $\widetilde{O}(q \cdot n^{1/2 + \gamma + \varepsilon})$ computation. 
\end{restatable}

We then apply landmark resampling to the problem of maintaining random walks in dynamic graphs. 

For directed graphs, we show how to maintain the \emph{PageRank} estimate of each node under dynamic edge insertions and deletions. In the static setting, the PageRank can be estimated by launching a number of random walks of random length distributed as $Geom(\lambda)$ for some $\lambda \in (0, 1)$ \cite{bahmani2010fast}. Our approach allows to maintain these walks ensuring that in dynamic setting, at any time, the estimated PageRank of every node remains bounded by its maximal historical value.

\begin{restatable}{theorem}{pageRankThm}
Using landmark resampling to maintain $k = \Omega(\log n)$ random walks of a random length distributed as $Geom(\lambda)$ from each node ensures that, with high probability, at any time $T$, the estimated PageRank of every node $v$ satisfies
\begin{gather*}
\widehat{\textnormal{PageRank}}^T(v) \leq \max_{t' \leq T} \textnormal{PageRank}^{t'}(v) \cdot O(\log T).
\end{gather*}
Here, $\textnormal{PageRank}^t(v)$ and $\widehat{\textnormal{PageRank}}^t(v)$ denote the true and the estimated PageRank of node $v$ at time $t$ respectively.
\end{restatable}

Analogously, the same principles can be used to maintain other random-walk–based statistics, such as SimRank~\cite{shao2015efficient}, or node embeddings such as node2vec \cite{grover2016node2vec}.

In the case of undirected graphs, we show how to maintain a collection of random walks such that, with high probability, every edge at every time step experiences low congestion. 

\begin{restatable}{theorem}{maintainLowCongestRandomWalksThm}
Consider a graph sequence $G_1, \ldots, G_T$ where $G_{i+1}$ is obtained from $G_i$ through adaptive adversarial edge deletions and insertions. Suppose every node $v$ in $G$ is a source of $k \leq \min\limits_t\deg_{G_t}(v)$ random walks of length $l$. Then, under landmark resampling, with high probability at time $T$ every edge has congestion at most $O(\log T\cdot\log n) \cdot l$.
\end{restatable}

Keeping congestion low is crucial for two reasons. First, in network applications, random walks are often used to establish communication paths; if many of these paths share the same edge, that edge becomes a bottleneck, slowing down traffic. Second, when an adversary deletes an edge, the algorithm must resample all walks passing through it; if many walks are affected, this results in high computational overhead per edge deletion.

\section{Settings and Framework}
In this section, we formally describe two particular settings we study in this work, and then give a generalized, abstract framework.
\subsection{Random Walks in Dynamic Graphs Setting}
\label{sec:random walks in dynamic graphs setting}
We consider a dynamic graph setting in which an adversary acts against an algorithm that maintains a collection of random walks over time. The process begins with an adversarially chosen graph $G_1$. Each node serves as the starting point for $k$ random walks, each of length $l$, which are initially sampled in $G_1$. At every time step $t = 1, 2, \ldots, T$, given the current graph $G_t$ and the current set of random walks, the adversary produces the next graph $G_{t+1}$ by adding or deleting edges. It may also designate a subset of random walks to be resampled in $G_{t+1}$. In addition, any random walk that was traversing an edge deleted from $G_t$ must be resampled. After these resamples, the algorithm may choose to resample an additional subset of random walks of its own choice. This process continues for $T$ time steps.

For a given time step $t$, we define the \emph{congestion} of an edge $e$ in $G_t$ as the number of random walks that traverse $e$. The maximum congestion at time $t$ is then a maximum edge congestion over all edges at time $t$.

\subsection{Job-Machine Setting}
\label{sec:job-machine setting}
In this section, we formally introduce the job-machine setting from \cite{bhattacharya2022simple}.

In it, we consider a collection of jobs and machines. Each job has several eligible routines - groups of machines on which it can be executed, but at any given time, it must be assigned to exactly one of these routines. This relationship can be naturally represented as a hypergraph, where jobs and machines are the vertices, and each hyperedge (routine) of the form $\{j, m_1, \ldots, m_k\}$ denotes that job $j$ can be simultaneously assigned to the machines $m_1, \ldots, m_k$. If the adversary deletes a machine that is currently used by a job, that job becomes invalid and must be resampled. Resampling a job means selecting uniformly at random one of the routines incident to that job in the current hypergraph. Now, we introduce the necessary notation.

Let $G_t$ be the job-machine hypergraph at time $t$. Let $J$ denote the set of all jobs and $M$ denote the set of all machines. Let $R_t$ be the set of all hyperedges (routines) of $G_t$, and $R_t(u)$ for $u \in J \cup M$ be the set of routines at time $t$ incident to $u$. Let $\Delta$ be the maximal degree over all jobs in $G_1$. Denote with $T$ the total number of steps in the process; since in \cite{bhattacharya2022simple} the adversary only advances to the next step by deleting a machine, $T \leq |M|$. Let $A_t$ be the assignment at time $t$, i.e., hyperedges used. Let $x^t$ be the machine deleted at time $t$, and $load^t(x)$ be the load of the machine $x$ at time $t$, i.e., the number of edges in $A_t$ incident to $x$. Let $target^t(x)$ be the expected load of $x$ at time $t$ if all the jobs are sampled simultaneously at time $t$, that is $target^t(x) = \sum\limits_{j \in J}\frac{|R_t(x) \cap R_t(j)|}{|R_{t}(j)|}$.

Since the algorithm must resample all the jobs that were using the deleted machine, the algorithm’s total recourse is defined as $\sum\limits_{t \in [T]}load^t(x^t)$ and the objective is to minimize this quantity. 

\subsection{General Framework}
\label{sec:framework}
We now describe our general framework. The process involves $n$ objects and unfolds over discrete time steps $t = 1, 2, \ldots, T$.

At the beginning, the adversary specifies an initial collection of distributions $(\mathcal{D}_1^1, \ldots, \mathcal{D}_n^1)$, one for each object, and samples all objects according to these distributions. Throughout the paper, we assume all the individual distributions to be over some common universe $\mathcal{U}$.

At each subsequent time step $t$, the system evolves as follows. Based on the entire history up to that point $\mathcal{H}^t$, consisting of previously chosen distributions and realized outcomes, the adversary determines (i) the next set of distributions $(\mathcal{D}_1^{t+1}, \ldots, \mathcal{D}_n^{t+1})$, and (ii) the subset of objects to be sampled at time $t+1$.

Formally, the adversary can be viewed as a pair of functions
\begin{gather*}
    f: \mathcal{H}^{t} \rightarrow (\mathcal{D}_1^{t+1}, \ldots, \mathcal{D}_n^{t+1})\ \text{ and }\ g:\mathcal{H}^{t} \rightarrow \mathcal{I}_{adv}^{t+1} \subset [n],
\end{gather*}

where $f$ specifies the next distributions and $g$ specifies which objects are sampled.

Acting after the adversary, the algorithm can also decide to sample a subset of objects. Hence, it can be represented by a function

\begin{gather*}
    A: (\mathcal{H}^{t+1}, \mathcal{D}_1^{t+1}, \ldots, \mathcal{D}_n^{t+1}) \rightarrow \mathcal{I}_{alg}^{t+1} \subset [n],
\end{gather*}
which, given the history up to time $t+1$, including the realizations of the objects sampled by the adversary at time $t+1$, and the distributions chosen by the adversary for time $t+1$, outputs the subset of objects to be sampled at time $t+1$.

This way, at any time $t$, each object $u$ has an \emph{individual realization} $\mathcal{R}_u^t$ that either originates from time $t$—if $u$ was sampled at that step by either the adversary or the algorithm—or persists from an earlier time otherwise. The \emph{joint realization} at time $t$ is defined as the collection of all individual realizations, $(\mathcal{R}_1^t, \ldots, \mathcal{R}_n^t)$.

Given the adversary and the algorithm, we can define a \emph{joint distribution} at time $t$, which is the distribution over all possible joint realizations at that time. We also define the \emph{static distribution} at time $t$ as the product distribution $(\mathcal{D}_1^t, \ldots, \mathcal{D}_n^t)$ obtained when all objects are sampled simultaneously at time $t$, and the \emph{static realization} as a sample from this static distribution.

The goal of the analysis is to characterize the joint distribution induced by this process after $T$ time steps.

\section{Temporal Selection}
\label{sec:temporal selection}
In this section, we describe temporal selection - an attack on stochastic dynamic systems that, if not addressed, can produce a joint distribution that is highly biased compared to any static one. We first demonstrate an application to dynamic graphs, where the attack can be used to over-congest a single edge/vertex with random walks and then adapt the same idea to the job–machine model, resulting in a strong counter-example to  \Cref{clm:incorrectJobMaschieneLoad}.

\subsection{Attack Random Walks in Dynamic Graphs}
Maintaining a collection of random walks is a cornerstone primitive in various applications, including routing \cite{tian2005randomwalk}, graph learning \cite{grover2016node2vec}, and assessing nodes' similarity \cite{jeh2002simrank}. A natural measure of quality for such a collection of random walks is the maximum edge congestion, which is known to be small if all the walks are sampled simultaneously; in particular, the following property holds for static distributions. 

\begin{theorem}[Folklore]
    \label{thm:low congest rws}
    Consider an undirected graph $G$. Assume each node $v \in V(G)$ initiates $k \leq deg(v)$ random walks of length $l$. Then for every edge $e \in E(G)$, the expected congestion of all random walks passing through $e$ is no more than $l$. Moreover, w.h.p., every edge has the congestion of $O(l\cdot \log n)$.
\end{theorem}

Yet, perhaps surprisingly, we show through temporal selection that Theorem \ref{thm:low congest rws} does not hold in dynamic graphs in the following sense: if the graph is dynamic and nodes can launch random walks at different time steps, there can be an edge with expected congestion exponentially larger than $l$. Formally,

\begin{restatable}{theorem}{badRWsDynGraphs}
    \label{thm:high congest dyn rws}
    There exists an adaptive adversarial strategy that produces a sequence of $T = O(n)$ undirected graphs $G_1, G_2, \ldots, G_T$ with $E(G_1) \supset E(G_2) \supset \ldots \supset E(G_T)$ and for each vertex $v \in V(G_1)$ chooses a \emph{single} time $t(v) \in [T]$ to sample a \emph{single} random walk of length $O(\log n)$ from $v$ in $G_{t(v)}$, such that there exists an edge $e \in E(G_T)$ with an expected congestion $\Omega(n^{\varepsilon})$ at time $T$ for some constant $\varepsilon > 0$.
\end{restatable}
\begin{proof}[Proof's core intuition.]
    The core idea is simple: because individual distributions can change over time, an adversary can schedule samples so that each random walk is sampled at a time when the probability of its ``bad'' realization is high. We illustrate this on a toy example.
    \begin{figure}[ht]
        \centering
        \begin{tikzpicture}[scale=0.7, every node/.style={font=\small}]
          \coordinate (v) at (0,0);
          \coordinate (top) at (0,1.6);
          \coordinate (a) at (-1.6,-1.3);
          \coordinate (b) at (0,-1.3);
          \coordinate (c) at (1.6,-1.3);
        
          \draw[thick] (v) -- (top) node[midway, right] {$e$};
          \draw[thick] (v) -- (a);
          \draw[thick] (v) -- (b);
          \draw[thick] (v) -- (c);
        
          \fill (v) circle (1.6pt);
          \fill (top) circle (1.6pt);
          \fill (a) circle (1.6pt);
          \fill (b) circle (1.6pt);
          \fill (c) circle (1.6pt);
        
          \node[above] at (top) {};
          \node[below left] at (a) {$a$};
          \node[below] at (b) {$b$};
          \node[below right] at (c) {$c$};
        \end{tikzpicture}
    \end{figure}
    Consider a graph on $5$ vertices as in the picture, and assume an adversary wants to hit an edge $e$ with a random walk of length $2$ from either of the vertices $a$, $b$, or $c$. If it launches all the random walks simultaneously, the probability of some random walk traversing $e$ is $\frac{1}{4} + (1 - \frac{1}{4}) \cdot \frac{1}{4} + (1 - \frac{1}{4})^2\cdot\frac{1}{4} = \frac{37}{64}$. However, if after launching a walk from $a$, the adversary sees that it failed, it can cut an edge to $a$, increasing the probability for a random walk from $b$ to hit $e$, and launch a random walk from $b$ after that. If that fails as well, the adversary can further cut an edge to $b$, increasing the probability for a random walk from $c$. Following this strategy, the adversary can achieve hitting probability of $\frac{1}{4} + (1 - \frac{1}{4})\cdot\frac{1}{3} + (1 - \frac{1}{4})\cdot (1 - \frac{1}{3})\cdot\frac{1}{2} = \frac{3}{4} > \frac{37}{64}$. 

    The full proof of the Theorem based on this idea can be found in Appendix \ref{sec:appendix:proofs:temporal selection in dynamic graphs}.
\end{proof}

Since the adversarial strategy of Theorem \ref{thm:high congest dyn rws} does not rely on resampling the same random walk multiple times, but rather samples each random walk once, it can be used to attack proactive resampling in the setting of dynamic graphs and random walks. 

\begin{restatable}{theorem}{attackProactiveResamplingDynGraphs}
    \label{thm:attack proactive resampling dyn graphs}
    Consider Random Walks in Dynamic Graphs with Proactive Resampling. For this setting, there exists a graph $G$ on $n$ vertices, where each vertex is a source of a single random walk of length $O(\log n)$, and an adaptive adversarial strategy for $G$, that within $O(n)$ time steps, induces a maximum expected congestion of $\Omega(n^\varepsilon)$ for some constant $\varepsilon > 0$. 

    Notably, this adversarial strategy does not add edges. 
\end{restatable}
\begin{proof}[Proof's core intuition.] 
    To leverage the result of Theorem \ref{thm:high congest dyn rws} against Proactive resampling, an adversary needs to ensure that if a random walk from vertex $v$ sampled at time $t(v)$ hits the target edge, it will not get resampled for a long enough period for other random walks to also hit the target edge and hence aggregate the congestion. 

    The trick to achieving this is rather simple. In the construction of a proof sketch of Theorem \ref{thm:high congest dyn rws}, an adversary samples a random walk from $a$ at time $1$, and if that fails, cuts the edge to $a$ and samples a random walk from $b$ at time $2$. To achieve stability of those samples for the time period of $T$ steps, one can do the following. Sample random walks at times $1$ and $2$ respectively, but without yet deleting any edges. Then wait for $W = 1 + 2 + 4 + \ldots 2^r > T$ time steps to ensure that the proactive resampling period for walks from $a$ and $b$ is at least $T$. Then, at time $W + 1$, a walk from $a$ will be resampled ``automatically'' by the proactive resampling algorithm, and if this sample fails to hit $e$, an adversary cuts vertex $a$ increasing the probability for a random walk from $b$, which will be sampled at time $W + 2$ by proactive resampling, to hit $e$. This way, we essentially replay the adversarial strategy for $a$ and $b$ but at times $W + 1$ and $W + 2$, ensuring that samples at these steps are stable.  

    We dedicate the full proof based on this idea to Appendix \ref{sec:appendix:proofs:temporal selection in dynamic graphs}.
\end{proof}

Let us now proceed to an application of the temporal selection attack to the Job and Machines setting. 

\subsection{Attack Job-Machines}
For the concepts and notation of the job-machine setting, we refer the reader to Section~\ref{sec:job-machine setting}.

The recourse of the algorithm is defined in \cite{bhattacharya2022simple} as
\begin{gather*}
    \sum\limits_{t\leq T}load^t(x^t),
\end{gather*}
and the goal is to minimize this quantity. The work of \cite{bhattacharya2022simple} approaches this by bounding the load at each time step $t$ in terms of the target load at time $t$. In particular, the key technical lemma in \cite{bhattacharya2022simple} asserts that when using proactive resampling, these two quantities differ only by polylogarithmic additive and multiplicative factors.

\incorrectJobMachineLoad*

Through the temporal selection, we show that this is not the case. In particular, the attack proceeds as follows. The adversary chooses a machine $S$ to overload. Initially, every job has a small probability of sampling a hyperedge containing $S$. Successively for each job $j$, the adversary does the following. First, by deleting most of the hyperedges incident to $j$ that do not contain $S$, it increases the chance of $j$ hitting $S$ when sampled. Then, it samples $j$. Afterward, it deletes all the hyperedges of $j$ incident to $S$ but the sampled one to decrease the probability of $j$ hitting $S$ in the future. This way, for any static distribution, at most one job has a high chance of hitting $S$. Hence, for any static distribution, the expected total number of hits is low. However, the actual number of sampled hyperedges containing $S$ by the end of the process is high. Let us now proceed to the formal counterexample.

\temporalSelectionJobMachines*
\begin{proof}
Consider $n$ jobs. For every job $j$, there are $n$ machines specific to this job: $a_1^j, \ldots, a_{n}^j$. Apart from job-specific machines, there is a special machine $S$. Additionally, there are $k = O(n^3)$ machines $B = \{b_1, \ldots, b_k\}$ that are used to skip time. Every job $j$ is contained in $n$ hyper edges: $\{j, a_1^j\}, \ldots, \{j, a_{n-1}^j\}, \{j, a_n^j,  S\}$.

We think of jobs as being arranged in $\sqrt{n}$ groups of size $\sqrt{n}$. In this setup, the adversary does the following. It triggers the resampling of jobs in the group $i$ at times $[(X + \sqrt{n})\cdot i + 1, (X+\sqrt{n})\cdot i + \sqrt{n}]$, where $X = \sqrt{n} -1 + n^{3/2}-n$. We explain the choice of $X$ in a few moments. In between triggering two consecutive groups, the adversary skips time by deleting machines in $B$. After triggering all the groups, the adversary waits, by deleting machines in $B$ at each step, for the proactive resampling period for job $1$ to become more than $n^3$. Denote with $t(1)$ the first time step such that job $1$ is sampled at time $t(1)$ and the next time job $1$ is scheduled is more than $t(1) + n^3$. Note that $t(1) = O(n^3)$. This way we ensure that jobs in group $i$ will be sampled at times $\text{GroupSampleInterval}(i) = [(X + \sqrt{n})\cdot i + t(1) + 1, (X + \sqrt{n}) \cdot i + t(1) + \sqrt{n}]$. In particular, the last job is sampled before the first job is sampled again, that is, before the time $T = t(1) + n^3$. Before the $i$-th group gets sampled at the $\text{GroupSampleInterval}(i)$, the adversary \emph{pretrims} it. In particular, it leaves each job $j$ in the group with only $\sqrt{n}$ hyper edges by deleting $n - \sqrt{n}$ machines other than $a_n^j$ that are not a part of a currently selected hyperedge. So pretrimming the group takes $\sqrt{n} \cdot (n - \sqrt{n}) = n^{3/2} - n$ time steps. After pretrimming, all the jobs in the group get sampled by the algorithm during the $\text{GroupSampleInterval}(i)$. During this phase, the adversary skips time by deleting machines in $B$. After the group gets sampled, the adversary \emph{posttrims} it. Namely, for every job $j$ in the group that did not select a hyper edge containing $S$, it removes this edge by deleting $a_n^j$.

Pretrimming, sampling, and posttrimming constitute processing a group. After a group $i$ is processed, the adversary starts to process a group $i + 1$. Now we can justify the choice of $X$ - it is the time period needed to posttrim the group $i$ and pretrim the group $i + 1$ before the group $i + 1$ gets sampled.

\begin{figure}[h]
    \centering
    \resizebox{\linewidth}{!}{%
    \begin{tikzpicture}[
      >=latex,
      x=1cm,
      block/.style={minimum height=6pt, rounded corners=1pt, draw=none},
      lab/.style={font=\small, rotate=35, anchor=west, align=center},
      brace/.style={decorate, decoration={brace, amplitude=4pt, mirror}},
      mlabel/.style={font=\small, inner sep=1pt}
    ]
    
    \definecolor{trigblue}{RGB}{104,141,189}
    \definecolor{pretrim}{RGB}{196,77,69}
    \definecolor{sample}{RGB}{120,170,110}
    \definecolor{posttrim}{RGB}{151,128,189}
    
    \def\ybase{0}      %
    \def\h{0.22}       %
    \def\Ltrig{1.2}
    \def\LgapA{1.4}
    \def\LtrigB{1.2}
    \def\LgapB{1.5}    %
    \def\LpreA{2.0}
    \def\LsamA{1.3}
    \def\LpostA{1.3}
    \def\LpreB{2.0}
    \def\LsamB{1.3}
    \def\LpostB{1.3}
    
    \coordinate (x0) at (0, \ybase);
    \coordinate (x1) at ($(x0) + (\Ltrig,0)$);
    \coordinate (x2) at ($(x1) + (\LgapA,0)$);
    \coordinate (x3) at ($(x2) + (\LtrigB,0)$);
    \coordinate (x4) at ($(x3) + (\LgapB,0)$);
    \coordinate (x5) at ($(x4) + (\LpreA,0)$);
    \coordinate (x6) at ($(x5) + (\LsamA,0)$);
    \coordinate (x7) at ($(x6) + (\LpostA,0)$);
    \coordinate (x8) at ($(x7) + (\LpreB,0)$);
    \coordinate (x9) at ($(x8) + (\LsamB,0)$);
    \coordinate (x10) at ($(x9) + (\LpostB,0)$);
    
    \draw[line width=0.6pt] (x0) -- (x10);
    
    \draw[fill=trigblue, block] (x0) rectangle ++(\Ltrig, \h);       %
    \draw[fill=trigblue, block] (x2) rectangle ++(\LtrigB, \h);      %
    \draw[densely dotted, line width=0.6pt] (x3) -- (x4);            %
    
    \draw[fill=pretrim,  block] (x4) rectangle ++(\LpreA,  \h);      %
    \draw[fill=sample,   block] (x5) rectangle ++(\LsamA,  \h);      %
    \draw[fill=posttrim, block] (x6) rectangle ++(\LpostA, \h);      %
    
    \draw[fill=pretrim,  block] (x7) rectangle ++(\LpreB,  \h);      %
    \draw[fill=sample,   block] (x8) rectangle ++(\LsamB,  \h);      %
    \draw[fill=posttrim, block] (x9) rectangle ++(\LpostB, \h);      %
    
    \node[lab] at ($(x0) + (0.1,0.45)$) {Trigger group 1};
    \node[lab] at ($(x1) + (0.2,0.45)$) {Skip time};
    \node[lab] at ($(x2) + (0.1,0.45)$) {Trigger group 2};
    \node[lab] at ($(x3) + (0.2,0.45)$) {Skip time};
    \node[lab] at ($(x4) + (0.1,0.45)$) {Pretrim group 1};
    \node[lab] at ($(x5) + (0.05,0.45)$) {Group 1\\ gets sampled};
    \node[lab] at ($(x6) + (0.05,0.45)$) {Posttrim group 1};
    \node[lab] at ($(x7) + (0.1,0.45)$) {Pretrim group 2};
    \node[lab] at ($(x8) + (0.05,0.45)$) {Group 2\\ gets sampled};
    \node[lab] at ($(x9) + (0.05,0.45)$) {Posttrim group 2};
    
    \def\down{0.55} %
    
    \draw[brace] ($(x0)  + (0,-\down)$) -- ($(x1)  + (0,-\down)$)
      node[mlabel, midway, below=4pt] {$\sqrt{n}$};
    
    \draw[brace] ($(x1)  + (0,-\down)$) -- ($(x2)  + (0,-\down)$)
      node[mlabel, midway, below=4pt] {$X$};
    
    \draw[brace] ($(x2)  + (0,-\down)$) -- ($(x3)  + (0,-\down)$)
      node[mlabel, midway, below=4pt] {$\sqrt{n}$};
    
    \draw[brace] ($(x3)  + (0,-\down)$) -- ($(x4)  + (0,-\down)$)
      node[mlabel, midway, below=4pt] {$O(n^{3})$};
    
    \draw[brace] ($(x4)  + (0,-\down)$) -- ($(x5)  + (0,-\down)$)
      node[mlabel, midway, below=4pt] {$n^{3/2}-n$};
    
    \draw[brace] ($(x5)  + (0,-\down)$) -- ($(x6)  + (0,-\down)$)
      node[mlabel, midway, below=4pt] {$\sqrt{n}$};
    
    \draw[brace] ($(x6)  + (0,-\down)$) -- ($(x7)  + (0,-\down)$)
      node[mlabel, midway, below=4pt] {$\sqrt{n}-1$};
    
    \draw[brace] ($(x7)  + (0,-\down)$) -- ($(x8)  + (0,-\down)$)
      node[mlabel, midway, below=4pt] {$n^{3/2}-n$};
    
    \draw[brace] ($(x8)  + (0,-\down)$) -- ($(x9)  + (0,-\down)$)
      node[mlabel, midway, below=4pt] {$\sqrt{n}$};
    
    \draw[brace] ($(x9)  + (0,-\down)$) -- ($(x10) + (0,-\down)$)
      node[mlabel, midway, below=4pt] {$\sqrt{n}-1$};
    
    \end{tikzpicture}

    }
\caption{A timeline of a process for two groups. Values under the braces denote how long a given phase lasts.}
\end{figure}
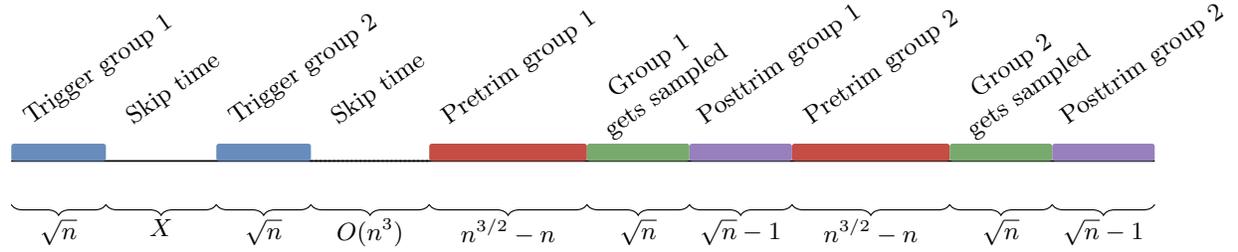

We now need to prove two claims. First, we state that the expected load of $S$ at time $T$ is $\sqrt{n}$. Second, we show that at any time the target load of $S$ is $O(1)$.

\noindent\textbf{Expected load of $S$.} Observe that once the group is pretrimmed, it has $\sqrt{n}$ jobs with each having $\sqrt{n}$ hyper edges with one of those containing $S$. Hence, the expected number of sampled hyperedges in this group that contain $S$ is $1$. Those remain after posttrimming. And since we have $\sqrt{n}$ groups, the expected total number of hyperedges containing $S$ at time $T$ is $\sqrt{n}$.

\noindent\textbf{Highest target load.} Initially, every job has $\frac{1}{n}$ probability to hit $S$, hence the expected load on $S$ from each group before pretrimming is $\frac{1}{\sqrt{n}}$. After a group gets pretrimmed, but before it gets posttrimmed, every node in this group has a probability of $\frac{1}{\sqrt{n}}$ to hit $S$, hence the expected load on $S$ from such a group is $1$. After a group gets posttrimmed, the expectation of load $L$ on $S$ from this group is $\E[L] = \E[\E[L \mid H]] = \E[H \cdot \frac{1}{\sqrt{n}} + (\sqrt{n} - H) \cdot 0]$, where $H$ is a random variable denoting the number of jobs in the group that hit $S$ in the sample phase. Since $\E[H] =1$, the expected load on $S$ from the group is $\frac{1}{\sqrt{n}}$. Now, since at any point in time all groups but one are either not yet pretrimmed or are already posttrimmed, the total expected load of $S$ at any point in time is no more than $(\sqrt{n} - 1)\cdot \frac{1}{\sqrt{n}} + 1 \leq 2$.

This way we get that 
\begin{gather*}
    \mathbb{E}[load^T(S)] = \sqrt{n}\ \ \text{ but }\ target^T(S) \leq 2
\end{gather*}
and more over $\max\limits_{t\leq T}target^t(S)\leq 2$ which disproves \Cref{clm:incorrectJobMaschieneLoad}.
    
\end{proof}

Overall, the temporal selection attack is a powerful adversarial approach to create biases in dynamic stochastic systems. In the next section, we incorporate the insights obtained from its analysis into the design of our algorithms.

\section{Algorithms}
\label{sec:algorithms}

Recall that in designing algorithms, we pursue two primary objectives: (i) the algorithm should perform only a small number of additional resamples, and (ii) the resulting joint distribution should remain close to a static one.

The high-level intuition behind our approach is that to approximate a static distribution, individual realizations should originate from only a few distinct time steps. Consequently, the joint distribution can be viewed as being “spliced” together from a small number of static ones.

To build intuition, consider a group assignment setting where the adversary aims to concentrate malicious participants in the first group. Suppose, moreover, that the adversary is restricted to sampling participants only at times~$1$ and~$2$. By Chernoff bounds, at each time step, not much more than $\frac{q}{g}$ malicious participants will fall into the first group, where $q$ is the total number of malicious participants and $g$ is the number of groups. Hence, even if the adversary constructs a final joint realization by splicing all malicious participants assigned to the first group at times~$1$ and~$2$, the total number still cannot exceed approximately $2 \cdot \frac{q}{g}$.

In this example, the adversary’s goal is to load the first group with malicious participants. Motivated by this intuition, we introduce a general notion of the \emph{load} of a joint realization.

\begin{definition}
    Consider a function $f: \mathcal{U}^n\rightarrow\mathbb{R}_{\geq 0}$. We call $f$ a \emph{load function} if and only if for all $(x_1, \ldots, x_n)\in \mathcal{U}^n$, $(y_1, \ldots, y_n)\in \mathcal{U}^n$, and $(z_1, \ldots, z_n)$ such that $z_i = x_i$ or $z_i = y_i$
    \begin{gather*}
        f(z_1, \ldots, z_n) \leq f(x_1, \ldots, x_n) +f(y_1, \ldots, y_n).
    \end{gather*}
\end{definition}
This definition captures the idea of forming a joint realization by ``splicing'' two static realizations together—selecting each individual realization from either the first or the second. Extending this operation to more than two static realizations yields the following load behavior.
\begin{lemma}
    \label{lem:compose loads}
    Let $f$ be a load function and consider $k$ input vectors to it: $\{(x_1^i, \ldots, x_n^i) \in \mathcal{U}^n\}_{i\in[k]}$ such that for every $i \in [k]$ $f(x_1^i, \ldots, x_n^i) \leq L$, then for $(z_1, \ldots, z_n) \in \mathcal{U}^n$ such that $z_i = x_i^j$ for some $j \in [k]$
    \begin{gather*}
        f(z_1, \ldots, z_n) \leq k\cdot L.
    \end{gather*}
\end{lemma}

We use the notion of a load function to formally reason about our algorithms: Greedy Temporal Aggregation and Landmark Resampling. 

\subsection{Greedy Temporal Aggregation}
The greedy temporal aggregation algorithm is gracefully simple. Denote with $objs^t(t')$ the set of objects whose realizations at time $t$ come from time $t'$. 

\begin{algorithm}
\caption{\texttt{Greedy Temporal Aggregation}}\label{alg:greedy ta}
\begin{algorithmic} 
    \State \textcolor{gray}{At each time step $t$:}
    \While{$\exists$ time steps $t_1 \neq t_2$ s.t. $|objs^t(t_1)|/|objs^t(t_2)| \in [1/2, 2]$}
      \State Sample all $objs^t(t_1)$
      \State Sample all $objs^t(t_2)$
    \EndWhile
\end{algorithmic}
\end{algorithm}

Yet, despite its simplicity, this algorithm achieves meaningful competitiveness with static distributions:
\GreedyTemporalAggregationThm*

In order to prove the load bound, we first demonstrate the following auxiliary lemma.
\begin{lemma}
    \label{lem:greedy temporal aggregation few time steps}
    When using greedy temporal aggregation, at time $T$ all realizations come from at most $O(\log n)$ time steps.
\end{lemma}
\begin{proof}
    Indeed, by the algorithm, if individual realizations come from different time steps $t_1$ and $t_2$, $|objs^T(t_1)|$ and $|objs^T(t_2)|$ differ by at least a factor of $2$ (otherwise those objects would have been resampled by the algorithm). Hence, at time $T$, there can be no more than $O(\log n)$ non-empty time steps, that is, $t'$ such that $objs^T(t') \neq \emptyset$. 
\end{proof}

The first part of the theorem is then proven as follows.
\begin{proof}[Proof of a low load.]
    Since each static realization has a load of no more than $L$ at each time step $t \in [T]$ with probability at least $1 - \varepsilon$, we conclude by the union bound that all static realizations have a load of at most $L$ with probability at least $1 - \varepsilon\cdot T$.

    By Lemma \ref{lem:greedy temporal aggregation few time steps}, the joint realization at time $T$ is spliced from $O(\log n)$ static realizations, and each of those has a load of at most $L$ with probability at least $1 - \varepsilon\cdot T$. Applying Lemma \ref{lem:compose loads} completes the proof.
\end{proof}

Second, we prove that the number of additional resamples the algorithm does is small.
\begin{proof}[Proof of few resamples.]
    To prove that the algorithm does a bounded number of additional samples, we conduct an amortized analysis. We say that each object has coins. When the adversary samples an object, we assume this object receives $\log_{3/2}n$ coins from the adversary. Moreover, if the object was previously sampled at time $t$ and there were $k$ objects sampled at time $t$, they all received $4/k$ coins.

    When an algorithm samples an object, it takes one coin from it.

    Denote with $group(u)$ a set of objects sampled at the same time step as $u$ and with $coins(u)$ the number of coins $u$ has. Consider the potential function of $u$: $\Phi(u) = \log_{3/2}(|group(u)|) + coins(u)$. We claim that this potential is monotonically increasing. Indeed, consider the cases when this potential is changed. 
    \begin{itemize}
        \item The object $u$ is sampled by the adversary. Then, the $|group(u)|$ may decrease (at most) from $n$ to $1$, hence the $\log_{3/2}(|group(u)|)$ from $\log_{3/2}(n)$ to $0$, but adversary gives the object $\log_{3/2}(n)$ coins, so $\Phi(u)$ can not decrease.
        \item Some other object from $group(u)$ is sampled by the adversary. Assume that before the sample $|group(u)| = k \geq 2$. Then $\Delta\Phi(u) = \log_{3/2}(k - 1) - \log_{3/2}(k) + 4/k > 0$.
        \item All the objects in $group(u)$ are sampled by the algorithm. This means that they are being sampled together with some other group, whose size is at least $\frac{1}{2}|group(u)|$. Therefore, the new size $|group'(u)| \geq |group(u)| + \frac{1}{2}|group(u)| = \frac{3}{2}|group(u)|$ and hence $\Delta\Phi(u) = \log_{3/2}(|group'(u)|) - \log_{3/2}(|group(u)|) - 1 \geq 0$.
    \end{itemize}
    
    This allows us to prove that the algorithm will always be able to take a coin. Indeed, initially, $\Phi(u) = \log_{3/2}n$. If the object is sampled, it means its group's size is increased by at least $3/2$ times, hence it was less than $\frac{n}{3/2}$, hence the number of coins was at least $1$, since $\log_{3/2}(|group(u)|) + coins(u) \geq \log_{3/2}n$ and $\log_{3/2}(|group(u)|) \leq \log_{3/2}n - 1$.

    Finally, note that 
    \begin{align*}
        \text{Total \# of samples by the alg} &\leq \text{\# of coins introduced to the system}\\
        &\leq \text{\# of adversarial samples}\cdot(\log_{3/2}n + 4).
    \end{align*}
\end{proof}

The greedy temporal aggregation algorithm can be naturally tuned to trade off the number of resamples it performs against the load guarantees.

\begin{algorithm}[H]
\caption{\texttt{Parameterized Greedy Temporal Aggregation}}\label{alg:gpta}
\begin{algorithmic} 
    \State \textbf{Parameter:} $\alpha \in [1, +\infty)$
    \State \textcolor{gray}{At each time step $t$:}
    \While{$\exists$ time steps $t_1 \neq t_2$ s.t. $|objs^t(t_1)|/|objs^t(t_2)| \in [1/\alpha, \alpha]$}
      \State Sample all $objs^t(t_1)$
      \State Sample all $objs^t(t_2)$
    \EndWhile
\end{algorithmic}
\end{algorithm}

This parameterized version gives the following guarantees.
\begin{restatable}{theorem}{ParameterizedGreedyTemporalAggregationThm}
    \label{thm:parameterized temporal aggregation}
    If the system evolves for $T$ time steps, and in each time step a static realization has a load at most $L$ with probability at least $1 - \varepsilon$, then Greedy Parameterized Temporal Aggregation with parameter $\alpha$ ensures that at time $T$ the joint realization has a load at most $L\cdot O(\log_\alpha n)$ with probability at least $1 - \varepsilon \cdot T$.

    Moreover, if the adversary does $q$ samples by the time $T$, Greedy Temporal Aggregation performs no more than $q\cdot O(\alpha\log n)$ samples.
\end{restatable}

The proof of this theorem is analogous to the proof of the non-parameterized version and is dedicated to Appendix~\ref{sec:appendix:proofs:parameterized ta}.

\subsection{Landmark Resampling}
The key idea of landmark resampling is to resample objects with an exponential back-off, like proactive resampling does, ensuring, however, that objects are being resampled at only a few different time steps called \emph{landmarks}. 

\begin{boxedremark}[Landmark Resampling]
    If an adversary samples an object at time $t_0$, landmark resampling schedules its resamples at times $t_1$, $t_2$, $t_3$, $\ldots$ with $t_i$ defined as $t_i = t_{i-1} + 2^i + x$, where $x \in [0,2^{i}]$ is chosen to maximize the number of trailing zeros in the binary description of $t_i$.
\end{boxedremark}

The landmark resampling enjoys the following guarantees:
\LandmarkResamplingThm*
We start the proof of the theorem by showing the following lemma.
\begin{lemma}
    \label{lem:few landmarks}
    The number of different time steps from which realizations might come from at a given time $T$ is $O(\log T)$. Moreover, the set of these time steps only depends on $T$.
\end{lemma}
\begin{proof}
    Denote a subset of objects that were resampled $i$ times by the time $T$ since last sampled by the adversary with $S_i$. Clearly, there are at most $\log_2 T$ non-empty $S_i$-s, since an object can be resampled $\log_2 T + 1$ times only after time $2^1 + 2^2 + \ldots +2^{\log_2 T + 1}> T$. We now show that realizations of objects from each $S_i$ come from constantly many different time steps at time $T$. 

    Consider an object $u \in S_i$ and define $t_j(u)$ for $j\leq i$ to be the time step at which $u$ was resampled for the $j$-th time. One can observe that $t_{i-1}(u) \in [T - 2^{i+3}, T]$. That is because $t_{i+1}(u) \leq t_i(u)+2^{i+2} \leq t_{i-1}(u) + 2^{i+1} + 2^{i+2} \leq t_{i-1}(u) + 2^{i+3}$. And hence, if $t_{i -1}(u) < T - 2^{i+3}$, then $t_{i+1}(u) < T$, which contradicts the fact that $u \in S_i$.  

    Now, notice that on every segment of natural numbers of length $2^i$, there will be a number with at least $i$ trailing zeros. This implies that $t_i(u)$ has at least $i$ trailing zeros for every $u \in S_i$. Furthermore, notice that on every segment of natural numbers of length $2^{i+3}$, there will be at most $8$ numbers with at least $i$ trailing zeros. Thus, since for all $u \in S_i$, $t_i(u) \in [T- 2^{i+3}, T]$, we conclude that for a given $i$, objects from $S_i$ are sampled at at most $8$ different time steps. 
\end{proof}

Denote the set of time steps from which individual realizations might come from at time $T$ with $landmarks(T)$.

The first part of the theorem is then proven as follows.
\begin{proof}[Proof of a low load.]
    By Lemma~\ref{lem:few landmarks}, the joint realization at time $T$ is spliced from $O(\log T)$ static realizations, each corresponding to one of the landmark time steps in $landmarks(T)$. For each such step, a static realization has a load at most $L$ with probability at least $1 - \varepsilon$. Applying a union bound over all $O(\log T)$ landmark steps, we obtain that all these static realizations simultaneously have load at most $L$ with probability at least $1 - \varepsilon \cdot O(\log T)$. Finally, by Lemma~\ref{lem:compose loads}, the splicing of these $O(\log T)$ realizations—namely, the joint realization at time $T$—also has load at most $L$, completing the proof.
\end{proof}

The proof of the second part of the theorem is relatively straightforward.
\begin{proof}[Proof of few resamples.]
    Fix some time $T$ and an object $u$. We show that for each adversarial sample of $u$, the landmark resampling schedules $O(\log T)$ samples of $u$ before $T$. 
    
    Let $t_0$ denote the time step at which an adversary samples $u$, and $t_1, t_2, t_3\ldots$ denote the time steps at which landmark resampling resamples $u$. Note that $t_i \geq t_0 + 2^i$, therefore, if $t_k$ is the last resampling before $T$, $k$ is $O(\log T)$.
\end{proof}

\subsection{Table Games}
We now present a unified and simplified perspective on the algorithms considered in this work, intended to make their comparison more transparent. Specifically, we recast the interaction between the algorithm and the adversary as a \emph{table game}.

In such a game, the adversary constructs a table with one row per object and one column for each time step $1,\ldots T$. The table is built column by column: for every column, the adversary first specifies a distribution for each object, then samples all objects according to these distributions and observes the realizations. After all $T$ columns have been constructed, the adversary chooses one realization for each object, obtaining the final joint realization. This completes the general formulation of a table game. To get a table game corresponding to a specific algorithm, we further constrain this general description by specifying which realizations the adversary may select in the end. 

We now define a corresponding table game for each of the algorithms under consideration. For every such game, we argue that the table game adversary possesses at least as much power as the adversary in the original framework in terms of which joint distributions they can obtain. Consequently, if a certain joint realization is unlikely to occur in the table game regardless of the adversary’s choices, it must also be unlikely to occur in the original setting. Although the adversary is formally stronger in the table games, these additional capabilities play only a minor role, and the table games can therefore be regarded as effectively \emph{equivalent} to the original framework.

\paragraph{Greedy Temporal Aggregation.}
In the greedy temporal aggregation table game, the adversary is restricted to selecting individual realizations originating from at most $O(\log n)$ distinct time steps.

We now show that this table game faithfully represents the original setting. In particular, we prove that the table game adversary $\mathcal{A}^{tbl}$ can reproduce the behavior of any adversary $\mathcal{A}^o$ in the original model.

\begin{theorem}
\label{thm:greedy temporal aggregation table game}
For every adversarial strategy $\mathcal{A}^o$ in the original setting with greedy temporal aggregation, there exists a corresponding strategy $\mathcal{A}^{tbl}$ in the greedy temporal aggregation table game that induces the same joint distribution.
\end{theorem}
\begin{proof}
The key idea is that $\mathcal{A}^{tbl}$ can simulate the actions of $\mathcal{A}^o$ step by step.

At each time step $t$, the original adversary $\mathcal{A}^o$ chooses a static distribution $\mathcal{D}^t = (\mathcal{D}_1^t, \ldots, \mathcal{D}_n^t)$ and observes realizations for some objects $ \mathcal{I}^t \subset [n]$ from it. Its decisions depend only on the history of its past observations, denoted $\mathcal{H}_{t-1}$. Therefore, formally, we can view $\mathcal{A}^o$ as a pair of functions $(f, g)$, where
\begin{gather*}
    \mathcal{D}^t = f(\mathcal{H}^{t-1})\ \text{ and }\ \mathcal{I}^t = g(\mathcal{H}^{t-1}).
\end{gather*}
We prove by induction on $t$ that $\mathcal{A}^{tbl}$ can produce the same distribution over histories $\mathcal{H}^t$ as $\mathcal{A}^o$.

For $t = 1$, both adversaries have empty history, so the claim trivially holds.
Assume now that at time $t-1$, $\mathcal{A}^{tbl}$ has already matched the distribution of $\mathcal{H}^{t-1}$ generated by $\mathcal{A}^o$. Since $\mathcal{A}^o$ selects $\mathcal{D}^t = f(\mathcal{H}^{t-1})$, and $\mathcal{A}^{tbl}$ can simulate $f$, both adversaries draw $\mathcal{D}^t$ from the same distribution $\sim f(\mathcal{H}^{t-1})$. Then, $\mathcal{A}^{tbl}$ obtains $\mathcal{H}^{t}$ by augmenting $\mathcal{H}^{t-1}$ with realizations of objects $g(\mathcal{H}^{t-1})$, exactly as $\mathcal{A}^o$ would. Thus, their resulting histories $\mathcal{H}^t$ have identical distributions.

By induction, $\mathcal{A}^{tbl}$ and $\mathcal{A}^o$ generate the same distribution for $\mathcal{H}^{T+1}$. By Lemma \ref{lem:greedy temporal aggregation few time steps}, the final joint realization obtained by $\mathcal{A}^o$ is composed of individual realizations from $\mathcal{H}^{T+1}$ that come from at most $O(\log n)$ different time steps. Therefore, $\mathcal{A}^{tbl}$—which observes the entire table and is allowed to select realizations from at most $O(\log n)$ distinct time steps—can pick exactly the same joint realization as $\mathcal{A}^o$. Hence, $\mathcal{A}^{tbl}$ is at least as powerful as $\mathcal{A}^o$, completing the proof.
\end{proof}

\paragraph{Landmark Resampling.}
In the landmark resampling table game, the adversary is restricted to selecting individual realizations from $landmarks(T)$ -- a fixed set of $O(\log T)$ landmark time steps, which depend only on $T$.

Similarly to greedy temporal aggregation, we show that the table-game adversary for landmark resampling is at least as powerful as its counterpart in the original setting.

\begin{restatable}{theorem}{landmarkTableGame}
\label{thm:landmark table game}
For every adversarial strategy $\mathcal{A}^o$ in the original setting with landmark resampling, there exists a corresponding strategy $\mathcal{A}^{tbl}$ in the landmark resampling table game that induces the same joint distribution.
\end{restatable}

The proof closely follows that of greedy temporal aggregation (Theorem~\ref{thm:greedy temporal aggregation table game}), differing only in that it applies Lemma~\ref{lem:few landmarks} in place of Lemma~\ref{lem:greedy temporal aggregation few time steps}. For completeness, we include the full proof in Appendix~\ref{sec:appendix:proofs:landmark resampling}.

\paragraph{Proactive Resampling.}
The proactive resampling table game differs slightly from the previous two. In this setting, after selecting the distributions for a given column but before observing their outcomes, the adversary may \emph{mark} a subset of objects in this column. At the end of the game, the adversary may choose individual realizations only among the marked ones. The marking is limited: for each object, the adversary can mark at most $O(\log T)$ realizations.

To understand why this corresponds to the original proactive resampling setup, recall that when an object is sampled by the adversary at time $t$, the algorithm schedules its resamples at times $t + 2^i$. Suppose, for instance, that an object is first sampled at time $1$; then it will be resampled at roughly time $T/2$. Consequently, neither the realization at time $1$ nor any realization between $1$ and $T/2$ is relevant, as these will be overwritten by the later resampling. In effect, to observe a relevant realization, the adversary must sample an object and then wait for about half of the remaining time—yielding at most $O(\log T)$ relevant realizations per object.

The table-game adversary is thus strictly stronger: it may choose the final realization from any $O(\log T)$ time steps, not necessarily exponentially spaced, and it makes this choice after observing all of them.

\section{Applications}
We now illustrate the versatility of our framework through three use cases. First, we demonstrate how parameterized greedy temporal aggregation can be used to employ existing palette sparsification and sublinear coloring algorithms in dynamic settings. Then, we revisit and refine the main result of “Simple Dynamic Spanners”~\cite{bhattacharya2022simple}, replacing its use of proactive resampling with greedy temporal aggregation, which restores the claimed guarantees. Finally, we show how landmark resampling can be used to efficiently maintain a collection of random walks and estimate PageRank in dynamic graphs.

\subsection{Dynamic Palette Sparsification and Coloring}

In the vertex coloring problem, the goal is to assign a color to each vertex of a graph so as to obtain a \emph{proper} coloring, meaning that no two adjacent vertices receive the same color. A classical fact is that any graph with maximum degree $\Delta$ admits a proper coloring using at most $\Delta+1$ colors.

A key refinement of this guarantee was shown by Assadi, Chen, and Khanna~\cite{assadi2019sublinear}. They proved that it is not necessary for every vertex to have access to the full set of $\Delta+1$ colors: instead, if each vertex is given a short random palette of only $O(\log n)$ colors sampled uniformly from $\{1,\dots,\Delta+1\}$, then with high probability the graph still has a proper $(\Delta+1)$-coloring that respects these per-vertex palettes.

\begin{theorem}[Palette-Sparsification Theorem of~\cite{assadi2019sublinear}]
\label{thm:palette-sparsification}
Let $G$ be an $n$-vertex graph with maximum degree $\Delta$. For each vertex $v \in V(G)$, sample a list $L(v)$ of $O(\log n)$ colors independently and uniformly at random from $\{1,\ldots,\Delta+1\}$. Then, with high probability, there exists a proper $(\Delta+1)$-coloring of $G$ in which every vertex $v$ receives a color from its list $L(v)$.
\end{theorem}

In a dynamic setting, however, the situation is more delicate. If an adaptive adversary is allowed to observe the sampled palettes and then perform edge insertions and deletions, the initially sampled palettes may eventually cease to support any proper coloring of the current graph. This necessitates selectively resampling palettes over time.

We show that parameterized greedy temporal aggregation provides a way to maintain palette sparsification against an adaptive adversary, at the cost of using a constant-factor larger global color range and incurring an $n^{\varepsilon}$ overhead in the number of palette resamples.

\begin{theorem}
    \label{thm:dynamic palettes}
    Consider a sequence of graphs on $n$ vertices $G_1, \ldots, G_T$ where $G_{t+1}$ is obtained from $G_t$ through adaptive adversarial edge deletions and insertions. Suppose that there exists an integer $\Delta$ such that for every $t \in [T]$, the maximum degree in $G_t$ is no more than $\Delta$. Then, for any $\varepsilon > 0$, using parameterized greedy temporal aggregation with parameter $n^{\varepsilon}$, for every vertex $v \in V(G)$ and time $t \in [T]$, we can maintain a palette $L_t(v)$ such that 
    \begin{itemize}
        \item For every $t \in [T]$, $G_t$ is proper colorable with high probability with each $v \in V(G)$ receiving a color from $L_t(v)$
        \item Each $L_t(v)$ is of size $O(\log n)$
        \item For every $t \in [T]$, the total number of colors used, i.e., $|\bigcup\limits_{v \in V(G)}L_t(v)|$, is $O(\frac{1}{\varepsilon}\cdot \Delta)$ 
        \item If the adversary performs $q$ edge insertions and deletions, the algorithm does $O(q\cdot n^\varepsilon\cdot \log n)$ palette resamples.
    \end{itemize}
\end{theorem}
\begin{proof}
    We begin by mapping a problem to our framework. For each vertex, its palette is treated as an object, and we say that the adversary resamples a palette of vertex $v$ iff it inserts an edge incident to $v$.

    When a subset of palettes is resampled, either by the adversary or the algorithm, at time $t$ we sample them from a fresh set of $\Delta + 1$ colors, which we call a \emph{range}. Sampling a palette from a range is picking $O(\log n)$ uniformly random colors from this range.
    
    This completes the description of the procedure. Let us now prove the stated properties.
    
    The fact that the graph is colorable at any time with high probability follows from two observations. First, vertices sampled at different time steps have distinct palettes since they are sampled from disjoint ranges, so no conflict is possible between two vertices sampled at different times no matter which colors they pick from their palettes. Second, vertices that are sampled at the same time induce a subgraph, which clearly has a max degree no more than $\Delta$, and hence is colorable with high probability with colors from palettes due to Theorem \ref{thm:palette-sparsification}.  
    
    Palettes have size $O(\log n)$ at any time by construction. 
    
    Lemma \ref{lem:parameterized ta few time steps} ensures that under parameterized temporal aggregation with parameter $n^\varepsilon$ at any point in time, realization of objects comes from at most $O(\frac{1}{\varepsilon})$ time steps, meaning at any point in time at most $O(\frac{1}{\varepsilon})$ ranges are used. The total number of colors potentially used at time $t$ is the number of colors in each range times the number of ranges, that is $(\Delta + 1) \cdot O(\frac{1}{\varepsilon})= O(\frac{1}{\varepsilon}\cdot \Delta)$. 
    
    Finally, Theorem \ref{thm:parameterized temporal aggregation} ensures that if adversary causes $q'$ object resamples, then with parameter $n^\varepsilon$, the parameterized temporal aggregation does at most $O(q' \cdot n^\varepsilon \cdot \log n)$ object resamples, and if in our context an adversary inserts no more than $q$ edges, it causes no more than $q' \leq 2q$ object resamples so the theorem claim follows.  
\end{proof}

A major consequence of the palette-sparsification phenomenon of Assadi, Chen, and Khanna is that it enables \emph{sublinear-time} randomized algorithms for $(\Delta+1)$-vertex coloring. In particular, they showed that one can compute a proper $(\Delta+1)$-coloring in time that is independent of the number of edges.

\begin{theorem}[Sublinear Coloring~\cite{assadi2019sublinear}]
\label{thm:sublinear coloring}
Let $G$ be an $n$-vertex graph with maximum degree $\Delta$. There exists a randomized algorithm that, with high probability, computes a proper $(\Delta+1)$-coloring of $G$ in $\widetilde{O}(n\sqrt{n})$ time.
\end{theorem}

Notably, whenever the input graph is sufficiently dense (e.g., $|E(G)| \gg n^{3/2}$), the algorithm runs in time sublinear in the input size.

Building on our framework, we further show how to \emph{maintain} a proper $O(\Delta)$ coloring under adaptive adversarial edge updates with sublinear average computation per update. More precisely, we obtain an amortized update cost of $\widetilde{O}(n^{1/2+\varepsilon})$, which is sublinear in the number of vertices.

\dynamicSublinearColoringThm*
\begin{proof}
    We frame the problem in our framework. In particular, a vertex is considered an object, and if an adversary inserts an edge $(u, v)$, objects $u$ and $v$ are considered adversarially resampled. We then use parameterized greedy temporal aggregation with parameter $n^\varepsilon$ on top of this. 

    When a subset of nodes is being resampled at time $t$, either adversarially or by the algorithm, these nodes induce a subgraph $S_t$. We launch an algorithm from Theorem \ref{thm:sublinear coloring} on $S_t$, giving it fresh $\Delta + 1$ colors to use. This completes the description of the algorithm. 

    Let us first show that this algorithm produces a proper coloring of $G$ at any time with high probability. By induction on $t$, we assume that $G_t \setminus S_t$ is properly colored, $S_t$ is properly colored whp by Theorem \ref{thm:sublinear coloring}, and there can not be conflict between two nodes in $G_t \setminus S_t$ and $S_t$ since $S_t$ uses a fresh set of $\Delta + 1$ colors. 

    The number of colors used at any time is bounded by $O(\frac{1}{\varepsilon}\cdot \Delta)$ since by Lemma \ref{lem:parameterized ta few time steps}, objects' realizations come from at most $O(\log_{n^\varepsilon}n) = O(\frac{1}{\varepsilon})$ different time steps, and for each time step we are using $\Delta + 1$ colors.

    As for the total computation, we first observe that the number of adversarially sampled objects is no more than $2q$ since each edge insertion only resamples two vertices. Now define $n_t = |V(S_t)|$ - the number of nodes being sampled at time $t$. By Theorem \ref{thm:sublinear coloring}, at time $t$, with high probability, the algorithm uses $\widetilde{O}(n_t\sqrt{n_t})$ computation, hence

    \begin{align*}
        \mathrm{Total\ Computation} &= \sum\limits_{t \in [T]}\widetilde{O}(n_t \sqrt{n_t})\\
        &\leq \sqrt{n}\cdot \sum\limits_{t \in [T]}\widetilde{O}(n_t)
    \end{align*}
    where the last inequality is because $\sqrt{n_t}\leq\sqrt{n}$.
    
    Furthermore, Theorem \ref{thm:parameterized temporal aggregation} gives us that $\sum\limits_{t \in [T]}n_t = O(q\cdot n^{\varepsilon} \cdot \log n)$. Thus we obtain
    \begin{align*}
        \mathrm{Total\ Computation} &= \sqrt{n}\cdot \widetilde{O}(q\cdot n^{\varepsilon}),
    \end{align*}
    which completes the proof.
\end{proof}

The palette-sparsification paradigm was subsequently strengthened by Alon and Assadi~\cite{beyondPalette} in the special case of \emph{triangle-free} graphs. Unlike the general setting—where $\Delta+1$ colors are always sufficient and essentially unavoidable—triangle-free graphs can be colored using asymptotically fewer than $\Delta$ colors. Alon and Assadi showed that this improved color bound is still compatible with sparsified random per-vertex palettes.

\begin{theorem}[Palette-Sparsification Theorem of \cite{beyondPalette}]
    \label{thm:palette-sparsification triangle free}
    Let $G(V, E)$ be an $n$-vertex triangle-free graph with maximum degree $\Delta$. Let $\gamma \in (0, 1)$. Suppose for any vertex $v \in V(G)$, we sample $O(\Delta^\gamma + \sqrt{\log n})$ colors $L(v)$ from a set of $O(\frac{\Delta}{\gamma\cdot\ln \Delta})$ colors independently and uniformly at random. Then, with high probability, there exists a proper coloring of $G$ in which the color for every vertex $v$ is chosen from $L(v)$.
\end{theorem}

As in \cite{assadi2019sublinear}, this insight can be used to obtain a sublinear algorithm for coloring of triangle-free graphs.

\begin{theorem}[Sublinear Triangle-Free Coloring of \cite{beyondPalette}]
    \label{thm:triangle-free coloring sublinear alg}
    Let $\gamma \in (0, 1)$. There exists a randomized algorithm that, with high probability, properly colors a triangle-free graph on $n$ vertices with maximum degree $\Delta$ into $O(\frac{\Delta}{\gamma \cdot \log \Delta})$ colors in time $\widetilde{O}(n^{3/2 + \gamma})$.
\end{theorem}

Our framework naturally supports ``dynamizing'' results of this type, extending them to the fully dynamic setting. We show that $O(\frac{\Delta}{\ln \Delta})$ palette sparsification \emph{and} sublinear coloring algorithm are both possible against an adaptive adversary.  

\begin{restatable}{theorem}{triangleFreeDynamicPaletteThm}
    \label{thm:dynamic tringle-free palettes}
    Consider a sequence of graphs on $n$ vertices $G_1, \ldots, G_T$ where $G_{t+1}$ is obtained from $G_t$ through adaptive adversarial edge deletions and insertions. Each $G_t$ is guaranteed to be triangle-free; furthermore, there exists an integer $\Delta$ such that for every $t \in [T]$, the maximum degree in $G_t$ is no more than $\Delta$. Then, for any $\gamma, \varepsilon \in (0, 1)$, using parameterized greedy temporal aggregation with parameter $n^{\varepsilon}$, for every vertex $v \in V(G)$ and time $t \in [T]$, we can maintain a palette $L_t(v)$ such that 
    \begin{itemize}
        \item For every $t \in [T]$, $G_t$ is properly colorable with high probability with each $v \in V(G)$ receiving a color from $L_t(v)$
        \item Each $L_t(v)$ is of size $O(\Delta^\gamma + \sqrt{\log n})$
        \item For every $t \in [T]$, the total number of colors used, i.e., $|\bigcup\limits_{v \in V(G)}L_t(v)|$, is $O(\frac{1}{\varepsilon}\cdot \frac{\Delta}{\gamma\cdot \ln \Delta})$ 
        \item If the adversary performs $q$ edge insertions and deletions, the algorithm does $O(q\cdot n^\varepsilon\cdot \log n)$ palette resamples.
    \end{itemize}
\end{restatable}

\dynamicSublinearTriangleFreeColoringThm*

The proofs follow those of Theorems \ref{thm:dynamic palettes} and \ref{thm:dynamic sublinear coloring} closely and are dedicated to Appendix \ref{sec:appendix:proofs:dynamic palettes}.

\subsection{Simple Dynamic Spanners~\cite{bhattacharya2022simple} via Temporal Aggregation}

Before proceeding, we recommend that the reader review the notation and definitions from \cite{bhattacharya2022simple}, which are summarized in Section~\ref{sec:job-machine setting}.

In the job-machine setting, when a machine gets deleted, all the jobs that were using it must be resampled. Hence, the total recourse of the algorithm is defined as
\begin{gather*}
    \sum\limits_{t\leq T}load^t(x^t),
\end{gather*}
and this quantity is stated to be low in \cite{bhattacharya2022simple}, namely, bounded by the number of jobs times some logarithmic factors.

The central argument supporting this statement in \cite{bhattacharya2022simple} is \Cref{clm:incorrectJobMaschieneLoad}. It states that, under proactive resampling, for every time step $t$ and machine $x$, $load^t(x) = \widetilde{O}(target^t(x))$. Then, the argument proceeds as follows:
\begin{gather*}
    \sum\limits_{t\leq T}load^t(x^t) = \widetilde{O}(\sum\limits_{t\leq T}target^t(x^t)) = \widetilde{O}(|J|).
\end{gather*}

As demonstrated in Section~\ref{sec:temporal selection}, \Cref{clm:incorrectJobMaschieneLoad} does not hold under proactive resampling, leaving a gap in the proof.

In this section, we show that replacing proactive resampling with greedy temporal aggregation yields a streamlined, self-contained proof of the main result of~\cite{bhattacharya2022simple}. We note, however, that alternative proof also appears in~\cite{fine2025minimizing}. Accordingly, the purpose of this section is not to present a new 
$3$-spanner algorithm, but rather to illustrate how our framework can be applied to the job–machine setting.

To recover the main result of~\cite{bhattacharya2022simple}, we make two key modifications to the original analysis.

First, we substitute proactive resampling with greedy temporal aggregation, which yields a weaker but sufficient guarantee: competitiveness not with the instantaneous target $target^t(x)$, but with the maximal historical one:
\begin{gather*}
    load^t(x) = \widetilde{O}(\max\limits_{t' \leq t}target^{t'}(x)).
\end{gather*} 
Second, to make this relaxation sufficient for the proof, we rely on the additional assumption that each routine contains at most two machines—an assumption that indeed holds in the specific setting of \cite{bhattacharya2022simple}. This yields the following sequence
\begin{gather*}
    \sum\limits_{t\leq T}load^t(x^t) = \widetilde{O}(\sum\limits_{t\leq T}\max\limits_{t' \leq t}target^t(x^t)) = \widetilde{O}(|J|).
\end{gather*}

The proof consists of two major parts. First, through greedy temporal aggregation, we show that the load at time $t$ is competitive with the maximum historical load:
\begin{lemma}
    \label{lem:loadWithTargetBound}
    With high probability, for any time $t \leq T$
    \begin{gather*}
        load^t(x^t) = O(\log |J| \cdot \log T) \cdot \max\limits_{t'\leq t}target^{t'}(x_t) + O(\log |J| \cdot \log T).
    \end{gather*}
\end{lemma}
\begin{proof}
    We begin by mapping the problem to our framework. Each job is treated as an object, and its realization corresponds to the routine this job is assigned to. Whenever the adversary deletes a machine that a given job uses, the corresponding object is seen to be adversarially resampled. Resampling an object means picking uniformly at random one of its eligible routines.

    Now, fix a time $t$ and consider a machine deleted at that time -- $x^t$. For an assignment $A$ define a function $f(A)$ to be the number of routines in $A$ that contain $x^t$. Clearly, $f$ is a load function: if $\mathcal{R}_1$ and $\mathcal{R}_2$ are two joint realizations, and $\mathcal{R}_3$ is formed by choosing for each job either its realization from $\mathcal{R}_1$ or from $\mathcal{R}_2$, then the number of routines in $\mathcal{R}_3$ containing $x^t$ is at most the sum of those containing $x^t$ in $\mathcal{R}_1$ and $\mathcal{R}_2$.

    Consider the static distribution $S_{\hat{t}}$ at some time $\hat{t} \leq t$. Suppose we sample jobs according to $S_{\hat{t}}$, and let $X_i$ be an indicator variable that equals $1$ if the routine of the $i$-th job contains $x^t$, and $0$ otherwise. Then
    \begin{gather*}
        \E\left[\sum\limits_iX_i\right] = target^{\hat{t}}(x^t).
    \end{gather*}
    
    Moreover, since all $X_i$ are independent, by standard Chernoff bounds, with probability at least $1 - \frac{1}{T^{c+1}}$ for an arbitrary large $c$
    \begin{gather*}
        f(S_{\hat{t}}) = \sum\limits_iX_i = O(\log T)\cdot target^{\hat{t}}(x^t) + O(\log T).
    \end{gather*}
    And hence with probability at least $1 - \frac{1}{T^{c+1}}$, $f(S_{\hat{t}})$ possesses a bound independent of $\hat{t}$:
    \begin{gather*}
        f(S_{\hat{t}}) = O(\log T)\cdot \max\limits_{t'\leq t}target^{t'}(x^t) + O(\log T).
    \end{gather*}
    Thus by Theorem \ref{thm:greedy temporal aggregation main}, with probability at least $1 - \frac{1}{T^c}$
    \begin{gather*}
        load^t(x^t) = O(\log |J| \cdot \log T) \cdot \max\limits_{t'\leq t}target^{t'}(x_t) + O(\log |J| \cdot \log T).
    \end{gather*}
\end{proof}

This completes the application of greedy temporal aggregation, and the rest of the proof is a modification of the original proof of \cite{bhattacharya2022simple}. In particular, given Lemma~\ref{lem:loadWithTargetBound}, it is sufficient to give the bound on the sum of maximal historical target loads.

\begin{restatable}{lemma}{lowTargetSum}
    \label{lem:low target sum}
    Given that each routine contains at most $2$ machines,
    \begin{gather*}
        \sum\limits_{t \in [T]}\max\limits_{t'\leq t}target^{t'}(x^t) = O(\log \Delta) \cdot |J|.
    \end{gather*}
\end{restatable}
\begin{proof}
We start by expanding the definition of the target load
\begin{align*}
    \sum\limits_{t \in [T]}\max\limits_{t'\leq t}target^{t'}(x^t) = \sum\limits_{t \in [T]}\max\limits_{t'\leq t}\sum\limits_{j \in J}\frac{|R_{t'}(j)\cap R_{t'}(x^t)|}{|R_{t'}(j)|}
\end{align*}
Swapping max and sum, we get 
\begin{align*}
    \sum\limits_{t \in [T]}\max\limits_{t'\leq t}target^{t'}(x^t) &\leq \sum\limits_{t \in [T]}\sum\limits_{j \in J}\max\limits_{t'\leq t}\frac{|R_{t'}(j)\cap R_{t'}(x^t)|}{|R_{t'}(j)|}\\
    &= \sum\limits_{j \in J}\sum\limits_{t \in [T]}\max\limits_{t'\leq t}\frac{|R_{t'}(j)\cap R_{t'}(x^t)|}{|R_{t'}(j)|}.
\end{align*}

We now bound the inner sum for any $j$. Intuitively, the $t$-th term in this inner sum is an upper bound on how much $j$ contributes to the load of $x^t$ at time $t$. This upper bound is obtained as a historical maximum of the expected load of $j$ on $x^t$, or in other words the expected load at time $t'(x^t)$ defined as $\arg\max\limits_{t'\leq t}\dfrac{|R_{t'}(j)\cap R_{t'}(x^t)|}{|R_{t'}(j)|}$. At that time, $j$ had a set of incident hyper edges that contained $x^t$, denote those $S_{t'(x^t)}$, and the expected load is a fraction of those over the total number of edges $j$ had at time $t'(x^t)$, denote those $U_{t'(x^t)}$. We are therefore interested in analyzing the sum
\begin{align*}
    \sum\limits_{t\in[T]}\frac{|S_{t'(x^t)}|}{|U_{t'(x^t)}|}.
\end{align*}
To do so, we make two observations. First, since the set of hyperedges of $j$ only undergoes deletions with time, the collection of sets $\mathcal{U} = \{U_{t'(x^t)}\}_{t\in[T]}$ can be ordered by inclusion. Second, since every hyperedge contains at most two machines, and a hyperedge $r$ appears in some $S_{t'(x^t)}$ only if $x^t \in r$, and each machine is deleted at most once, we deduce that every hyperedge can appear in at most two $S$-s. Now, letting $U = R_1(j)$, so that $|U| \leq \Delta$, we apply Lemma \ref{lem:charging jobs} (see below) to get
\begin{align*}
    \sum\limits_{t \in [T]}\max\limits_{t'\leq t}\frac{|R_{t'}(j)\cap R_{t'}(x^t)|}{|R_{t'}(j)|} = O(\log \Delta),
\end{align*}
which allows us to conclude that 
\begin{align*}
    \sum\limits_{t \in [T]}\max\limits_{t'\leq t}target^{t'}(x^t) = O(\log \Delta) \cdot |J|
\end{align*}
completing the proof.
\end{proof}

\begin{restatable}{lemma}{chargingJobs}
\label{lem:charging jobs}
Let $U$ be a set with $|U|=n$, and let $\mathcal{U}$ be a collection of $m$ nested subsets of $U$. Index subsets in $\mathcal{U}$ in a way that $U_1 \supseteq \ldots \supseteq U_m$, and for each $i \in [m]$, let $S_i$ be some subset of $U_i$. Also, assume that every $u\in U$ belongs to at most two of the sets $S_i$. Then,
\begin{align*}
    \max \sum_{i=1}^m \frac{|S_i|}{|U_i|} \;=\; O(\log n).
\end{align*}
\end{restatable}
We prove this lemma in Appendix \ref{sec:appendix:proofs:rectifying}.

Combined together, Lemmas \ref{lem:loadWithTargetBound} and \ref{lem:low target sum} result in the original low-recourse claim of \cite{bhattacharya2022simple}:
\begin{gather*}
    \sum\limits_{t\in[T]}load^t(x^t) = O(\log \Delta \cdot \log |J| \cdot \log T) \cdot |J|.
\end{gather*}

\subsection{Random Walks in Dynamic Graphs}
Before proceeding, we recommend that the reader review the relevant notation and definitions, which are summarized in Section~\ref{sec:random walks in dynamic graphs setting}.

Sampling multiple random walks from each node is a common building block in a wide range of graph algorithms. For instance, it is used to compute node embeddings in network representation learning methods such as node2vec~\cite{grover2016node2vec}, or to measure node similarity in algorithms like SimRank~\cite{jeh2002simrank}.

When the underlying network is dynamic, however— i.e., when edges are deleted or added—maintaining such random walks becomes challenging. In particular, the deletion of a single edge invalidates all walks that traverse it, requiring them to be resampled. In highly available systems, where updates must be processed quickly and with minimal downtime, this can lead to significant computational overhead. An adversary could exploit this vulnerability by forcing many walks to pass through the same edge (like in Theorem~\ref{thm:high congest dyn rws}) and then deleting it, thereby triggering excessive resampling effort.

To mitigate this issue, we employ our landmark resampling algorithm that provides strong congestion guarantees:

\maintainLowCongestRandomWalksThm*
\begin{proof} 
    We begin by mapping the problem to our framework. Each random walk, labeled by $(v, i)$ for $v \in V(G)$ and $i \in [k]$, is treated as an object, and its realization corresponds to the set of edges it traverses. Whenever the adversary deletes an edge that a given random walk uses, the corresponding object is seen to be adversarially resampled.

    Now, fix an edge $e$. Define $f(\mathcal{R})$ for a joint realization $\mathcal{R}$ as the number of random walks in $\mathcal{R}$ that traverse $e$. Clearly, $f$ is a load function: if $\mathcal{R}_1$ and $\mathcal{R}_2$ are two joint realizations, and $\mathcal{R}_3$ is formed by choosing for each random walk either its realization from $\mathcal{R}_1$ or from $\mathcal{R}_2$, then the number of walks in $\mathcal{R}_3$ that traverse $e$ is at most the sum of those traversing $e$ in $\mathcal{R}_1$ and $\mathcal{R}_2$. 

    Consider the static distribution $S_{\hat{t}}$ at some time $\hat{t} \leq T$. Suppose we sample random walks according to $S_{\hat{t}}$, and let $X_i$ be an indicator variable that equals $1$ if the $i$-th random walk goes through $e$ and $0$ otherwise. Then, by Theorem \ref{thm:low congest rws} $f(S_{\hat{t}})$ possesses a bound independent of $\hat{t}$:
    \begin{gather*}
        f(S_{\hat{t}}) = \sum\limits_iX_i = O(\log n)\cdot l
    \end{gather*}
    
    with probability at least $1 - n^{-(c+2)}$, for any constant $c$. Applying Theorem~\ref{thm:landmark resampling} then yields that, at any time $T$, the number of random walks traversing $e$ is at most 
    \[ O(\log T\cdot \log n)\cdot l \] 
    with probability at least $1 - O(\log T) / n^{c+2}$. Taking a union bound over all edges completes the proof. 
\end{proof}

Along with congestion guarantees, landmark resampling preserves the mixing behavior of random walks. To demonstrate this, we provide a theorem stating that for every vertex $v$ of a graph, an adversary is not able to make all the random walks from $v$ end up in a specific region $S_v$ of the graph, provided $S_v$ is small enough. 
\begin{theorem}
    Consider a graph sequence $G_1, \ldots, G_T$ where $G_{i+1}$ is obtained from $G_i$ through adaptive adversarial edge deletions and insertions. Further assume that for every $i \in [T]$, the mixing time of $G_i$ is upper bounded by some integer $l$. Assume that for each vertex $v \in V(G_1)$, in the beginning, an adversary can select a subset $S_v \subset V(G_1)$ such that $|S_v| = O(n/\log T)$.

    Then there exists an integer $k = O(\log n)$ such that every node in the graph is a source of $k$ random walks of length $l$ and under landmark resampling, no matter the adversarial strategy, with high probability at time $T$ for every node $v \in V(G_1)$ there exists a random walk starting at $v$ whose endpoint is not in $S_v$. 
\end{theorem}
\begin{proof}
    Lemma \ref{lem:few landmarks} states that there exists a constant $c$ such that at time $T$, realizations of random walks come from at most $c \cdot \log_2 T$ time steps. Assuming for each $i \in [T]$, a random walk of length $l$ mixes in $G_i$, and assuming $|S_v| \leq \frac{n}{c\log_2 T}$ for every vertex $v$, we conclude that the probability for a random walk to end up in $S_v$ if sampled in a particular time step is no more than $\frac{1}{c\log_2 T}$. 

    By Theorem \ref{thm:landmark table game}, the probability of a random walk ending in $S_v$ at time $T$ is no more than the probability of it ending in $S_v$ at one out of $c\log_2T$ time steps, that is no more than $1 - (1 - \frac{1}{c\log_2T})^{c\log_2T} \leq 1 - \frac{1}{e}$. 

    The expected number of random walks from $v$ ending up in $S_v$ at time $T$ is then no more than $k\cdot (1 - \frac{1}{e})$. Let $k = 20\log n$. Since all random walks are independent, by Chernoff bounds, all random walks from $v$ end up in $S_v$ with probability no more than $O(\frac{1}{n^{2}})$. Taking a union bound over all vertices, we conclude that with probability $O(\frac{1}{n})$, for any vertex $v$, at least one random walk out of $k$ from $v$ does not end up in $S_v$ at time $T$.
\end{proof}

\subsection{PageRank}
The PageRank of a node in a directed graph on $n$ vertices is defined as the stationary distribution of the following random walk: with constant probability $\lambda$, the walk jumps to a uniformly random node in the graph, and with probability $1 - \lambda$, it moves to a uniformly random out-neighbor of the current node.

In practice, to estimate the PageRank of each node, one can run $k = \Omega(\log n)$ independent random walks from every node of a random length distributed as $Geom(\lambda)$. Given a collection of such random walks, the \emph{estimated PageRank} of a node $v$ is defined as 
\begin{gather*}
    \widehat{\textnormal{PageRank}}(v) = \frac{H_v}{nk/\lambda},
\end{gather*}
where $H_v$ is the total number of walks that pass through $v$. The work of \cite{bahmani2010fast} shows that this quantity is sharply concentrated around the actual PageRank of $v$.

In dynamic graphs, however, the situation becomes more delicate. If an adversary can resample different random walks at different times, the system becomes vulnerable—for example, the adversary may perform a temporal selection attack (analogous to that in Theorem~\ref{thm:high congest dyn rws}) to overload a specific node. In the PageRank context, this would correspond to a page acquiring a PageRank value substantially higher than it ever had at any single moment in time.

To prevent such bias, we employ landmark resampling, which yields the following guarantee.

\pageRankThm*
\begin{proof}
     We begin by mapping the problem to our framework. Each random walk, labeled by $(v, i)$ for $v \in V(G)$ and $i \in [k]$, is treated as an object, and its realization corresponds to the set of edges it traverses. Whenever the adversary deletes an edge that a given random walk uses, the corresponding object is seen to be adversarially resampled.

    Now, fix a node $v$. For a joint realization $\mathcal{R}$, let $H_v(\mathcal{R})$ be the number of random walks in $\mathcal{R}$ passing through $v$. Further, define $f(\mathcal{R}) = \frac{H_v(\mathcal{R})}{nk/\lambda}$. Clearly, $f$ is a load function: if $\mathcal{R}_1$ and $\mathcal{R}_2$ are two joint realizations, and $\mathcal{R}_3$ is formed by choosing for each random walk either its realization from $\mathcal{R}_1$ or from $\mathcal{R}_2$, then the number of walks in $\mathcal{R}_3$ that pass through $v$ is at most the sum of those passing through $v$ in $\mathcal{R}_1$ and $\mathcal{R}_2$. And the denominator $nk/\lambda$ is independent of $\mathcal{R}$.
    
    Consider the static distribution $S_{\hat{t}}$ at some time $\hat{t} \leq T$. Then the sharp-concentration result from \cite{bahmani2010fast} ensures that
    \[
        f(S_{\hat{t}}) =O(\textnormal{PageRank}^{\hat{t}}(v))
    \]
    with probability at least $1 - n^{-(c+1)}$, for any constant $c$. And hence with probability at least $1 - n^{-(c+1)}$, $f(S_{\hat{t}})$ possesses a bound independent of $\hat{t}$:
    \begin{gather*}
        f(S_{\hat{t}}) = O(\max\limits_{t'\leq T}\textnormal{PageRank}^{t'}(v))
    \end{gather*}

    Applying Theorem~\ref{thm:landmark resampling} then yields that at time $T$ the estimated PageRank of $v$ is at most
    \begin{gather*}
        O(\max\limits_{t'\leq T}\textnormal{PageRank}^{t'}(v) \cdot\log T)
    \end{gather*}
    with probability at least $1 - n^{-(c+1)}$. Taking a union bound over all nodes completes the proof.
\end{proof}

\newpage
\bibliographystyle{plain}
\bibliography{refs}

\newpage
\appendix
\section{Proofs}
\label{sec:appendix:proofs}

\subsection{Cuckoo analysis}
\label{sec:appendix:cuckoo}
Let us first recap the setting. There are $n$ participants and $g$ groups. At any point in time, each participant in the system must be assigned to some group. Participants might leave and join the system; however, they don't have the identities, so the algorithm can not recognize which of the left participants is rejoining. A fraction $p$ of participants can be malicious, and can leave and join the system adversarially. The algorithm can not tell if the participant is malicious or not. Initially, all the participants are assigned uniformly at random.

The cuckoo rule with parameter $k$ is the rule on what to do when the participant joins the system. In particular, when a participant $u$ joins, cuckoo assigns them to a uniformly random group, kicks out uniformly random $k$ participants from that group, and reasigns each of them to a uniformly random group, all that done independently. 

We now give two lemmas showing respectively that cuckoo maintains some properties of the purely random assignment, but not all. To this end, we utilize the assumption that at any point in time, the size of each group is between $\frac{n}{1.5g}$ and $\frac{1.5n}{g}$, which happens with high probability if all honest (not malicious) participants are assigned uniformly at random and $p\leq \frac{1}{6}$.

The following lemma demonstrates that the cuckoo rule is well-tailored to keep the honest majority in every group.
\begin{lemma}
    For $n$ participants, $g \leq n^{0.99}$ groups, if at most the fraction $p \leq 0.5 - \frac{1.1}{k} \lesssim 0.166$ of participants is malicious, then with high probability, over the course of polynomially many steps, the cuckoo rule with $k \geq 3$ ensures that in all groups the majority of the participants is honest.
\end{lemma}
\begin{proof}[Proof sketch]
    Fix some group, w.l.o.g., let's call it the first group. Define an \emph{epoch} to be the time period between one of the following scenarios: (I) the joining participant joins the first group, (II) the malicious participant gets kicked out and gets reassigned to the first group. 

    In scenario I, the expected gain of the first group in the number of malicious participants is $1$ (plus the negligible term for kicked out malicious participants to rejoin) minus the expected number of kicked out malicious participants, which is $p'k$, where $p'$ is a fraction of malicious participants in the first group. As for scenario II, the expected gain is just $1$ (plus the negligible term for other kicked out malicious participants to fall into the first group).

    The probability that an epoch terminates with scenario (I) at a given step is $\frac{1}{g}$ and with scenario (II) it is $1 - (1 - \frac{p}{g})^k \approx \frac{pk}{g}$. Hence, if $\Pr[\text{scenario I} \mid \text{end of epoch}] = \mathcal{P}$, then $\Pr[\text{scenario II} \mid \text{end of epoch}] = pk\cdot \mathcal{P}$. But these probabilities sum to $1$, so $\mathcal{P} = \frac{1}{1 + pk}$. 

    Now assume that the system has reached a state where the first group has a $p' = p + \frac{1}{k} + 0.01$ fraction of malicious participants. This is a necessary state on the way to obtaining a malicious majority. We want to analyze the expected net gain of the malicious participants in this first group in this regime: 
    \begin{align*}
        \E[\text{net gain} \mid \text{end of epoch}] &= \E[\text{net gain} \mid \text{scenario I}] \cdot \Pr[\text{scenario I} \mid \text{end of epoch}]\\ &+ \E[\text{net gain} \mid \text{scenario II}] \cdot \Pr[\text{scenario II} \mid \text{end of epoch}]\\
        &\lesssim (1 - p'k)\cdot\frac{1}{1 + pk} + 1\cdot\frac{pk}{1 + pk}\\
        &= 1 - \frac{p'k}{1 + pk},
    \end{align*}
    which is a negative quantity by the choice of $p'$. Define $\mu = -(1 - \frac{p'k}{1 + pk}) \approx 0.02$ - this is an expected loss in the considered regime of $p' = p + \frac{1}{k} + 0.01$ - the regime of consideration till the end of the proof. 

    Define $X_t$ to be the number of malicious participants in the first group after the $t$-th epoch. By what we have shown above, $\E[X_t \mid X_{t-1}, \ldots, X_0] \leq X_{t-1} - \mu$. Also, observe that for every $t$, $X_{t + 1} - X_t \leq k$ and $X_t - X_{t + 1} \leq k$, implying $|X_{t+1} - X_t|\leq k$.
    
    Now define $M_t = X_t + \mu\cdot t$. Observe that $M_t$ is a super-martingale:
    \begin{align*}
        \E[M_t \mid M_{t-1}, \ldots, M_0] &= \E[X_t\mid X_{t -1}, \ldots, X_0] + \mu\cdot t\\
        &\leq X_{t-1} - \mu +  \mu\cdot t\\
        &= M_{t-1}.
    \end{align*}

    Define $c$ to be $k + \mu \leq 2k$, then $|M_{t+1} - M_t|\leq c$ for all $t$.

    We assume that the size of the first group is always between $\frac{n}{1.5g}$ and $\frac{1.5n}{g}$. Hence, the adversary needs to gain at least $\frac{n}{1.5g}(0.5 - p') \geq 0.01\frac{n}{g}$ malicious participants in the negative gain regime.  We conclude by applying Azuma–Hoeffding's inequality to $M_t$:
    \begin{align*}
        \Pr[X_T> 0.5s]&=\Pr[X_T - X_0\geq 0.01\frac{n}{g}]\\
        &= \Pr[M_T - M_0\geq 0.01\frac{n}{g} + \mu \cdot T]\\
        &\leq \exp\left(\dfrac{(0.01\frac{n}{g} + \mu \cdot T)^2}{8Tk^2}\right)
    \end{align*}
    which is exponentially small for any choice of $T$.
\end{proof}
The following lemma shows that the cuckoo rule fails to maintain some pseudo-random properties, such as a uniform partition of a fixed group of participants. 
\begin{lemma}
    For $n$ participants, $g \leq n^{0.99}$ groups and cuckoo rule with $k \geq 3$, an adversary can gather a fixed set of $\frac{n}{2kg}$ participants in a fixed group at a fixed (linear) time.
\end{lemma}
\begin{proof}
    Denote the set of participants to gather $Q$ and call the group the adversary wants to aggregate them in the first group. We reference participants in $Q$ and only them as malicious within the scope of this lemma. The adversarial strategy is straightforward: while there exists a participant $q \in Q$ that is not assigned to the first group, leave and join $q$. 

    We prove under the assumption that the size of the first group, which we denote $s$, always satisfies $\frac{n}{1.5g}\leq s\leq \frac{1.5n}{g}$.

    We claim that within $T = \frac{12n}{g}$ steps, the adversary succeeds. Define an \emph{epoch} to be the time period between two events of one of the following types: (I) the joining malicious participant joins the first group, (II) the malicious participant gets kicked out by the cuckoo and gets reassigned to the first group. 
    
    We do an analysis of the net difference of the number of malicious participants in the first group after each epoch. If the epoch ends by the event (II), the net difference is then at least $+1$ and at most $+k$. If the event (I) occurs instead, the first group gets $1$ malicious participant by definition, but $k$ random ones are being kicked out. The probability for a single kicked out participant to be malicious is no more than $\frac{|Q|}{s}$, so the expected number of kicked malicious participants is no more than $\mu' = k\cdot \frac{|Q|}{s} \leq k\cdot\frac{\frac{n}{g}\cdot\frac{1}{2k}}{\frac{n}{1.5g}} = \frac{3}{4}$. Define $\mu = -\mu' + 1 \geq \frac{1}{4}$ -- this is the lower bound on the expected net gain of malicious participants in the first group. 

    Define $X_t$ to be the number of malicious participants in the first group after the $t$-th epoch. By what we have shown above, $\E[X_t \mid X_{t-1}, \ldots, X_0] \geq X_{t-1} + \mu$. Also, observe that $X_{t + 1} - X_t \leq k$ and $X_t - X_{t + 1} \leq k$, implying $|X_{t+1} - X_t|\leq k$.
    
    Now define $M_t = X_t - \mu\cdot t$. Observe that $M_t$ is a sub-martingale:
    \begin{align*}
        \E[M_t \mid M_{t-1}, \ldots, M_0] &= \E[X_t\mid X_{t -1}, \ldots, X_0] - \mu\cdot t\\
        &\geq X_{t-1} + \mu -  \mu\cdot t\\
        &= M_{t-1}.
    \end{align*}

    We rewrite the desired probability in terms of $M$:
    \begin{align*}
        \Pr[X_T \leq \frac{n}{g}\cdot\frac{1}{2k} ] &\leq \Pr[X_T \leq 2\frac{n}{g}]\\ 
        &\leq \Pr[X_T  \leq \mu\cdot T - \frac{n}{g}]\\ 
        &= \Pr[M_T \leq -\frac{n}{g}]\\
        &\leq \Pr[M_T - M_0 \leq -\frac{n}{g}].\\
    \end{align*}

    Define $c$ to be $k + \mu \leq 2k$, then $|M_{t+1} - M_t|\leq c$ for all $t$. We obtain the statement of the lemma by applying Azuma–Hoeffding's inequality to $M_T$:
    \begin{align*}
        \Pr[X_T \leq \frac{n}{g}\cdot\frac{1}{2k} ]&\leq \exp\left(\dfrac{-(\frac{n}{g})^2}{2Tc^2}\right)\\
        &= \exp\left(\Theta\left(\frac{n}{g}\right)\right).
    \end{align*}
\end{proof}

\subsection{Parameterized Greedy Temporal Aggregation}
\label{sec:appendix:proofs:parameterized ta}
\begin{lemma}
    \label{lem:parameterized ta few time steps}
    When using parameterized greedy temporal aggregation with parameter $\alpha$, at time $T$ all realizations come from at most $O(\log_\alpha n)$ time steps.
\end{lemma}
\begin{proof}
    Indeed, by the algorithm, if individual realizations come from different time steps $t_1$ and $t_2$, $|objs^T(t_1)|$ and $|objs^T(t_2)|$ differ by at least a factor of $\alpha$ (otherwise those objects would have been resampled by the algorithm). Hence, at time $T$, there can be no more than $O(\log_\alpha n)$ non-empty time steps, that is, $t'$ such that $objs^T(t') \neq \emptyset$. 
\end{proof}
\ParameterizedGreedyTemporalAggregationThm*
\begin{proof}
    \noindent\textbf{Low load.} Since each static realization has a load of no more than $L$ at each time step $t \in [T]$ with probability at least $1 - \varepsilon$, we conclude by the union bound that all static realizations have a load of at most $L$ with probability at least $1 - \varepsilon\cdot T$.

    By Lemma \ref{lem:parameterized ta few time steps}, the joint realization at time $T$ is composed of $O(\log_\alpha n)$ static realizations, and each of those has a load of at most $L$ with probability at least $1 - \varepsilon\cdot T$. Applying Lemma \ref{lem:compose loads} yields that the joint realization then has a load of at most $L\cdot O(\log_\alpha n)$.
    
    \noindent\textbf{Number of samples.} We convey an amortized analysis; we say that each object has coins. When the adversary samples an object, we assume this object receives $\log_{1 + 1/\alpha}n$ coins from the adversary. Moreover, if the object was previously sampled at time $t$ and there were $k$ objects sampled at time $t$, they all receive $\frac{2}{\ln(1 + 1/\alpha)\cdot k}$ coins.

    When an algorithm samples an object, it takes one coin from it.

    Denote with $group(u)$ a set of objects sampled at the same time step as $u$ and with $coins(u)$ the number of coins $u$ has. Consider the potential function of $u$: $\Phi(u) = \log_{1 + 1/\alpha}(|group(u)|) + coins(u)$. We claim that this potential is monotonically increasing. Indeed, consider the cases when this potential is changed. 
    \begin{itemize}
        \item The object $u$ is sampled by the adversary. Then, the $|group(u)|$ may decrease (at most) from $n$ to $1$, hence the $\log_{1 + 1/\alpha}(|group(u)|)$ from $\log_{1 + 1/\alpha}(n)$ to $0$, but adversary gives the object $\log_{1 + 1/\alpha}(n)$ coins, so $\Phi(u)$ can not decrease.
        \item Some other object from $group(u)$ is sampled by the adversary. Assume that before the sample $|group(u)| = k \geq 2$. Then $\Delta\Phi(u) = \log_{1 + 1/\alpha}(k - 1) - \log_{1 + 1/\alpha}(k) + \frac{2}{\ln(1 + 1/\alpha)\cdot k}$. Let us show that this quantity is positive. Indeed, 
        \begin{align*}
            &\log_{1 + 1/\alpha}(k - 1) - \log_{1 + 1/\alpha}(k) + \frac{2}{\ln(1 + 1/\alpha)\cdot k}\\
            &=\frac{\ln(k - 1)}{\ln(1 + 1/\alpha)} - \frac{\ln(k)}{\ln(1 + 1/\alpha)} + \frac{2}{\ln(1 + 1/\alpha)\cdot k}\\
            &= \frac{1}{\ln(1 + 1/\alpha)}\cdot(\ln(k - 1) - \ln(k) + \frac{2}{k})\\
            &\geq \frac{1}{\ln(1 + 1/\alpha)}\cdot (-\frac{1}{k - 1} + 2k).
        \end{align*}
        The first term is always positive, the second term is positive for $k \geq 2$.
        \item All the objects in $group(u)$ are sampled by the algorithm. This means that they are being sampled together with some other group, whose size is at least $\frac{1}{\alpha}|group(u)|$. Therefore, the new size $|group'(u)| \geq |group(u)| + \frac{1}{\alpha}|group(u)| = (1 + 1/\alpha)|group(u)|$ and hence $\Delta\Phi(u) = \log_{1 + 1/\alpha}(|group'(u)|) - \log_{1 + 1/\alpha}(|group(u)|) - 1 \geq 0$.
    \end{itemize}
    
    This allows us to prove that the algorithm will always be able to take a coin. Indeed, initially, $\Phi(u) = \log_{1 + 1/\alpha}n$. If the object is sampled, it means its group's size is increased by at least $1 + 1/\alpha$ times, hence it was less than $\frac{n}{1 + 1/\alpha}$, hence the number of coins was at least $1$, since $\log_{1 + 1/\alpha}(|group(u)|) + coins(u) \geq \log_{1 + 1/\alpha}n$ and $\log_{1 + 1/\alpha}(|group(u)|) \leq \log_{1 + 1/\alpha}n - 1$.

    Finally, note that 
    \begin{align*}
        \text{Total \# of samples by the alg} &\leq \text{\# of coins introduced to the system}\\
        &\leq \text{\# of adversarial samples}\cdot(\log_{1 + 1/\alpha}n + \frac{2}{\ln(1 + 1/\alpha)})\\
        &= O(\text{\# of adversarial samples}\cdot\frac{\ln(n) + 2}{\ln(1+1/\alpha)})\\
        &= O(\text{\# of adversarial samples}\cdot\alpha \cdot \log n)
    \end{align*}
\end{proof}

\subsection{Landmark Resampling}
\label{sec:appendix:proofs:landmark resampling}
\landmarkTableGame*
\begin{proof}
The key idea is that $\mathcal{A}^{tbl}$ can simulate the actions of $\mathcal{A}^o$ step by step.

At each time step $t$, the original adversary $\mathcal{A}^o$ chooses a static distribution $\mathcal{D}^t = (\mathcal{D}_1^t, \ldots, \mathcal{D}_n^t)$ and observes realizations for some objects $ \mathcal{I}^t \subset [n]$ from it. Its decisions depend only on the history of its past observations, denoted $\mathcal{H}^{t-1}$. Therefore, formally, we can view $\mathcal{A}^o$ as a pair of functions $(f, g)$, where
\begin{gather*}
    \mathcal{D}^t = f(\mathcal{H}^{t-1})\ \text{ and }\ \mathcal{I}^t = g(\mathcal{H^{t-1}}).
\end{gather*}
We prove by induction on $t$ that $\mathcal{A}^{tbl}$ can produce the same distribution over histories $\mathcal{H}^t$ as $\mathcal{A}^o$.

For $t = 1$, both adversaries have empty history, so the claim trivially holds.
Assume now that at time $t-1$, $\mathcal{A}^{tbl}$ has already matched the distribution of $\mathcal{H}^{t-1}$ generated by $\mathcal{A}^o$. Since $\mathcal{A}^o$ selects $\mathcal{D}^t = f(\mathcal{H}_{t-1})$, and $\mathcal{A}^{tbl}$ can simulate $f$, both adversaries draw $\mathcal{D}^t$ from the same distribution $\sim f(\mathcal{H}^{t-1})$. Then, $\mathcal{A}^{tbl}$ obtains $\mathcal{H}^{t}$ by augmenting $\mathcal{H}^{t-1}$ with realizations of objects $g(\mathcal{H}^{t-1})$, exactly as $\mathcal{A}^o$ would. Thus, their resulting histories $\mathcal{H}^t$ have identical distributions.

By induction, $\mathcal{A}^{tbl}$ and $\mathcal{A}^o$ generate the same distribution for $\mathcal{H}^{T+1}$. By Lemma \ref{lem:few landmarks}, the final joint realization obtained by $\mathcal{A}^o$ is composed of individual realizations from $\mathcal{H}^{T+1}$ that come from at most $O(\log T)$ different time steps and the set of this time steps only depends on $T$. Therefore, $\mathcal{A}^{tbl}$—which observes the entire table and is allowed to select realizations from exactly these $O(\log T)$ time steps—can pick the same joint realization as $\mathcal{A}^o$. Hence, $\mathcal{A}^{tbl}$ is at least as powerful as $\mathcal{A}^o$, completing the proof.
\end{proof}

\subsection{Temporal Selection in Dynamic Graphs}
\label{sec:appendix:proofs:temporal selection in dynamic graphs}
\badRWsDynGraphs*
\begin{proof}

    Our $G_1$ will be a complete ternary tree with root having an additional fourth edge $e$, which the adversary will attempt to congest. The ternary tree has the depth of $h = h_1 + h_2$, we choose $h_1$ and $h_2$ later in the analysis, but those will both be $\Theta(h) = \Theta(\log n)$. We call a random walk \emph{monotone} if it only goes up in the tree. The adversarial strategy will be to congest $e$ with monotone random walks of length $l = h + 1 = O(\log n)$ - one from each tree rooted at depth $h_1$. If it succeeds in doing so, the congestion on $e$ will be $3^{h_1} = n^\varepsilon$ walks in total.

    \begin{figure}[h]
        \centering
        \includegraphics[width=0.5\linewidth]{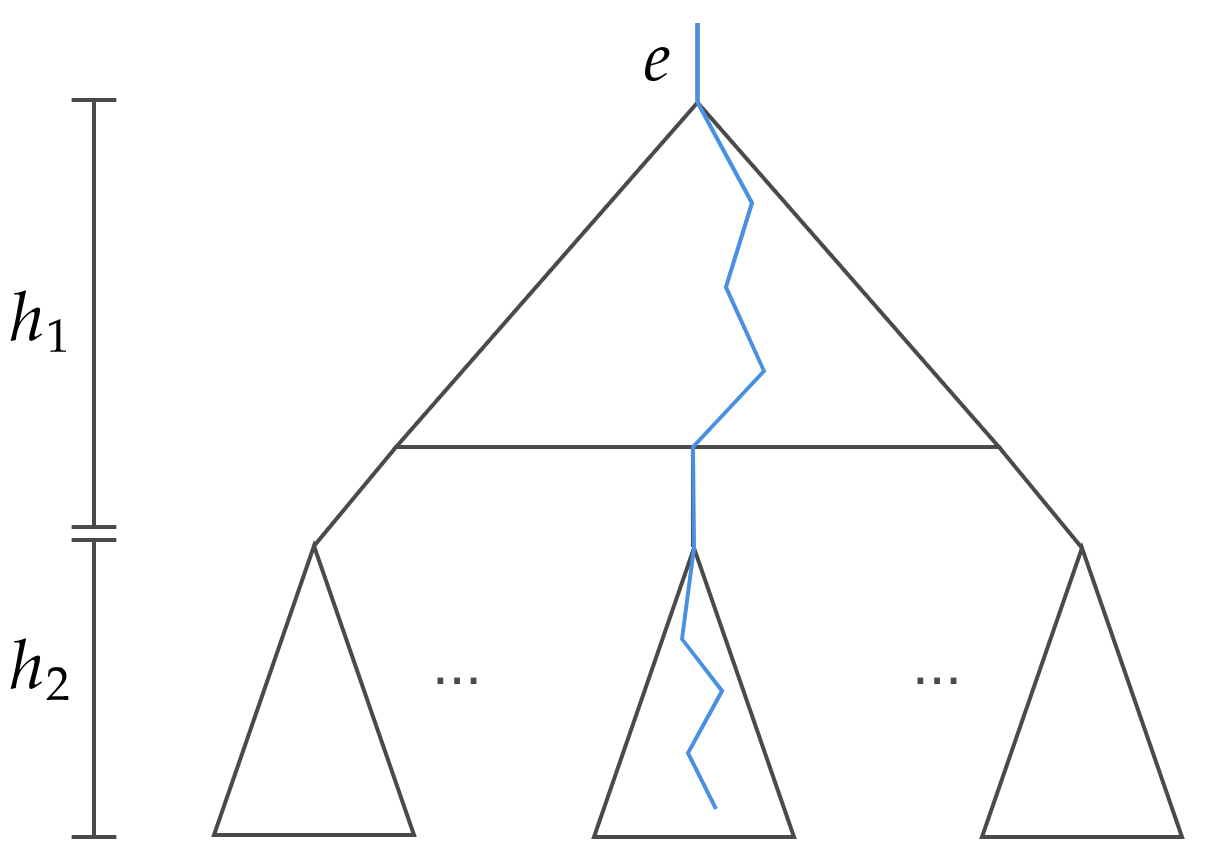}
        \caption{The upper part of the ternary tree and its subtrees at depth $h_1$. The blue monotone random walk starts at one of the subtrees and hits $e$.}
        \label{fig:placeholder}
    \end{figure}

    For each subtree at depth $h_1$, with root $r$ we apply the \texttt{RecourseDelete}$(r)$ procedure (Algorithm \ref{alg:recourse delete}). 
    
    \begin{algorithm}[ht]
        \caption{\texttt{RecourseDelete}$(v)$}
        \label{alg:recourse delete}
        \begin{algorithmic}
            \For{$child \in children(v)$}
                \State $success \gets RecourseDelete(child)$
                \If{$success$}
                    \State \textbf{return} $success$
                \EndIf
            \EndFor
            \State Launch a random walk from $v$
            \If{random walk is monotone and hits $e$}
                \State \textbf{return} \textbf{$success$}
            \Else 
                \State \textbf{delete} $(v, parent(v))$
                \State \textbf{return} $fail$
            \EndIf
        \end{algorithmic}
    \end{algorithm}

    Let us prove that such a strategy suffices. Let $t(v)$ be the time the random walk from $v$ is sampled. Let $p_v$ be the probability of a random walk from $v$ to be monotone and to hit $e$ if sampled at time $t(v)$. Now, fix some subtree $T$ at depth $h_1$ and for $v \in V(T)$ define $X_v$ to be an event of a random walk from $v$ being sampled, monotone, and hitting $e$. We are interested in analyzing $\Pr[\bigvee\limits_{v\in V(T)}X_v]$. First, denoting vertices of $T$ with $v_1, \ldots, v_k$ in the reverse order they are cut off by the RecourseDelete procedure (so $v_1$ is the root, $v_k$ is the left-most leaf), we rewrite
    \begin{align*}
        \Pr[\bigvee\limits_{v\in V(T)}X_v] &= \Pr\left[\bigvee\limits_{i \in [k]}\left(X_{v_i}\wedge\bigwedge\limits_{j \in [i+1, k]}\bar{X}_{v_j}\right)\right]\\
        &= \sum\limits_{i\in [k]}\Pr[X_{v_i}\wedge\bigwedge\limits_{j \in [i+1, k]}\bar{X}_{v_j}]\\
        &= \sum\limits_{i\in [k]}\Pr[X_{v_i} \mid \bigwedge\limits_{j \in [i+1, k]}\bar{X}_{v_j}] \cdot \Pr[\bigwedge\limits_{j \in [i+1, k]}\bar{X}_{v_j}]\\
    \end{align*}

    Now note that $\Pr[X_{v_i} \mid \bigwedge\limits_{j \in [i+1, k]}\bar{X}_{v_j}]$ is exactly $p_{v_i}$. Further, if we denote $Y = \bigvee\limits_{v\in V(T)}X_v$, then $\bigwedge\limits_{j \in [i+1, k]}\bar{X}_{v_j} \supset \bar{Y}$ for every $i \in [k]$, hence, $\Pr[\bigwedge\limits_{j \in [i+1, k]}\bar{X}_{v_j}] \geq \Pr[\bar{Y}] = 1 - \Pr[Y]$. 
    We can now rewrite
    \begin{align*}
        \Pr[Y] \geq \sum\limits_{i \in [k]}p_{v_i}\cdot (1 - \Pr[Y]).
    \end{align*}

    We want to show that $\Pr[Y]$ is at least $\frac{1}{2}$. Assume it is not. Then, 
    \begin{align*}
         \Pr[Y] \geq \frac{1}{2}\sum\limits_{i \in [k]}p_{v_i}.
    \end{align*}

    Consider some non-leaf vertex $v$ and its children $a$, $b$, $c$ (in the order they are cut off by the RecourseDelete). Note that $p_{a} = \frac{1}{4}p_v$, $p_{b} = \frac{1}{3}p_{v}$, and $p_{c} = \frac{1}{2}p_v$. Hence, sum of $p$-s of children is $\frac{13}{12}$ times the $p$ of the parent. Therefore, sum of the $p$-s of all leaves is $(\frac{13}{12})^{h_2}\cdot p_{v_1}$, and $p_{v_1} = (\frac{1}{4})^{h_1+1}$. Thus, setting $h_2 = 100h_1$, we get that 
    \begin{align*}
        \Pr[Y] \geq \frac{1}{2}\cdot (\frac{13}{12})^{100h_1}\cdot(\frac{1}{4})^{h_1 + 1} = \frac{1}{8}\cdot100^{h_1} > 1,
    \end{align*}

    Which is a contradiction. 

    What is left to observe is that there are $3^{h_1} = 3^{\frac{1}{101}\log_3n} = n^{\frac{1}{101}}$ subtrees at depth $h_1$. And with each one independently having more than $1/2$ chance to successfully execute RecourseDelete$(r)$ procedure, we conclude by Chernoff bounds that there will be $\widetilde{\Theta}(n^{\varepsilon})$ random walks that hit $e$ with $\varepsilon = \frac{1}{101}$. 
\end{proof}

\attackProactiveResamplingDynGraphs*
\begin{proof}
We reuse the construction and notation from the proof of Theorem~\ref{thm:high congest dyn rws}. In particular, let $A$ be the initial graph from that theorem (the ternary-tree gadget with the distinguished edge $e$), and fix the depth $h_1$. For every rooted subtree $T$ of $A$ whose root is at depth $h_1$, enumerate its leaves from right to left as
\[
v_1(T),v_2(T),\ldots,v_k(T).
\]
Let $T_i$ denote the trimmed state of $T$ at the moment $t(v_i)$ in the adversarial process of Theorem~\ref{thm:high congest dyn rws}, i.e., the remaining subtree of $T$ right when the walk from $v_i(T)$ is sampled there.

\medskip
\noindent\textbf{Graph $G$.}
Start from $A$ and add two gadgets.
(i) For every leaf $v_i(T)$ attach a fresh clique $C_{v_i(T)}$ of size $n^3$, connecting every clique vertex to $v_i(T)$.
(ii) Add a disjoint path component $B$ of length $W:=1 + 2 + 4 + \ldots + 2^r$ where $2^r\ge n$ and $W=\Theta(n)$.

\medskip
\noindent\textbf{Adversarial strategy.}
For $i=1,2,\ldots,k$, at time step $t=i$ the adversary samples (in parallel, over all such $T$) one length-$l=O(\log n)$ random walk from every $v_i(T)$. Then it deletes one edge of $B$ per time step for the next $W$ steps (“waits $W$ steps”).

By proactive resampling, each walk sampled at time $i$ is resampled at times $i+2^0,i+2^1,\ldots$. By definition of $W$, the first proactive resampling time after $W$ is $W+i$, and the \emph{next} one after that is at least
\[
(W+i) + 2^{r+1} \;\ge\; (W+i)+n,
\]
so there is a window of length $\Omega(n)$ after time $W+i$ during which this walk is not proactively resampled again.

Now we “replay” the trimming process from Theorem~\ref{thm:high congest dyn rws}, but aligned to these proactive-resampling times: for each $i=1,\ldots,k$ and each subtree $T$, if walks from $v_1, \ldots, v_{i-1}$ were not successfull so far, by time $W+i$ the adversary deletes edges in $T$ so that $T$ is in state exactly $T_i$ (the state in which Theorem~\ref{thm:high congest dyn rws} samples the walk from $v_i(T)$). An adversary also removes $C_{v_i(T)}$. Then at time $W+i$ the proactive resampling samples a fresh walk from $v_i(T)$ \emph{in the graph where $T$ is trimmed to $T_i$}.

\medskip
\noindent\textbf{Why early samples do not interfere.}
Let $X$ be the event that every walk sampled at times $1,\ldots,k$ moves into its attached clique on the first step and never returns to the leaf within its $l=O(\log n)$ steps at each out of $O(\log n)$ proactive resamples. Observer that $\Pr[\neg X]=O(n\log^2 n/n^3)=O(1/n)$ by a union bound. Conditioned on $X$, none of these early walks traverses any edge of $A$, hence none is invalidated by later deletions within $A$; in particular, the proactive resampling at time $W+i$ indeed samples from $v_i(T)$ in the intended trimmed tree $T_i$.

\medskip
\noindent\textbf{Congestion.}
Conditioned on $X$, within each subtree $T$ the process at times $W+1,\ldots,W+k$ is distributed exactly as in Theorem~\ref{thm:high congest dyn rws}: the walk from $v_i(T)$ is sampled precisely when the subtree is in state $T_i$, and the same subsequent trimming rule is applied. Therefore the same analysis implies that, in expectation, $\widetilde{\Omega}(n^{\varepsilon_0})$ subtrees contribute a successful walk that hits the edge $e$, yielding
\[
\E[\mathrm{cong}(e)\mid X]\ge \widetilde{\Omega}(n^{\varepsilon_0}).
\]
Unconditioning and using $\Pr[X]\ge 1-O(1/n)$ gives $\E[\mathrm{cong}(e)]\ge \widetilde{\Omega}(n^{\varepsilon_0})$.

Finally, $|V(G)|=N=O(n^4)$ due to the $n^3$-cliques, so $\widetilde{\Omega}(n^{\varepsilon_0})=\Omega(N^{\varepsilon})$ for $\varepsilon=\varepsilon_0/4>0$, completing the proof.

\paragraph{Remark on a modeling detail.}
We implicitly used the fact that if a walk is resampled due to proactive scheduling at time $t$, and simultaneously becomes invalid due to an edge deletion at time $t$, the proactive resampling period does not get reset. This is needed when the walk from $v_i$ is resampled at time $W + i$, and the adversary removes $C_{v_i}$, in which, conditioned on $X$, the walk was. This assumption can be lifted at a cost of a more technical analysis. 
\end{proof}

\subsection{Dynamic Palette Sparsification}
\label{sec:appendix:proofs:dynamic palettes}

\triangleFreeDynamicPaletteThm*
\begin{proof}
    We begin by mapping a problem to our framework. For each vertex, its palette is treated as an object, and we say that the adversary resamples a palette of vertex $v$ iff it inserts an edge incident to $v$.

    When a subset of palettes is resampled, either by the adversary or the algorithm, at time $t$ we sample them from a fresh set of $O(\frac{\Delta}{\gamma\cdot \ln \Delta})$ colors (see \cite{beyondPalette} for precise constants), which we call a \emph{range}. Sampling a palette from a range is picking $O(\Delta^\gamma + \sqrt{\log n})$ uniformly random colors from this range.
    
    This completes the description of the procedure. Let us now prove the stated properties.
    
    The fact that the graph is colorable at any time with high probability follows from two observations. First, vertices sampled at different time steps have distinct palettes since they are sampled from disjoint ranges, so no conflict is possible between two vertices sampled at different times no matter which colors they pick from their palettes. Second, vertices that are sampled at the same time induce a subgraph, which clearly has a max degree no more than $\Delta$ and is triangle-free, and hence is colorable with high probability with colors from palettes due to Theorem \ref{thm:palette-sparsification triangle free}.  
    
    Palettes have size $O(\Delta^\gamma + \sqrt{\log n})$ at any time by construction. 
    
    Lemma \ref{lem:parameterized ta few time steps} ensures that under parameterized temporal aggregation with parameter $n^\varepsilon$ at any point in time, realization of objects comes from at most $O(\frac{1}{\varepsilon})$ time steps, meaning at any point in time at most $O(\frac{1}{\varepsilon})$ ranges are used. The total number of colors potentially used at time $t$ is the number of colors in each range times the number of ranges, that is $O(\frac{\Delta}{\gamma\cdot \ln \Delta}) \cdot O(\frac{1}{\varepsilon})= O(\frac{1}{\varepsilon}\cdot \frac{\Delta}{\gamma\cdot \ln \Delta})$. 
    
    Finally, Theorem \ref{thm:parameterized temporal aggregation} ensures that if adversary causes $q'$ object resamples, then with parameter $n^\varepsilon$, the parameterized temporal aggregation does at most $O(q' \cdot n^\varepsilon \cdot \log n)$ object resamples, and if in our context an adversary inserts no more than $q$ edges, it causes no more than $q' \leq 2q$ object resamples so the theorem claim follows.  
\end{proof}

\dynamicSublinearTriangleFreeColoringThm*
\begin{proof}
    We frame the problem in our framework. In particular, a vertex is considered an object, and if an adversary inserts an edge $(u, v)$, objects $u$ and $v$ are considered adversarially resampled. We then use parameterized greedy temporal aggregation with $\alpha= n^\varepsilon$ on top of this. 

    When a subset of nodes is being resampled at time $t$, either adversarially or by the algorithm, these nodes induce a subgraph $S_t$. We launch an algorithm from Theorem \ref{thm:triangle-free coloring sublinear alg} on $S_t$, giving it fresh $O(\frac{\Delta}{\gamma \cdot \ln \Delta})$ colors to use (see \cite{beyondPalette} for precise constants).  

    It is straightforward to see that this algorithm produces a proper coloring of $G$ at any time with high probability. By induction on $t$, we assume that $G_t \setminus S_t$ is properly colored, $S_t$ is properly colored whp by the algorithm of Theorem \ref{thm:triangle-free coloring sublinear alg}, and there can not be conflict between two nodes in $G_t \setminus S_t$ and $S_t$ since $S_t$ uses a fresh color range. 

    The number of colors used at any time is bounded by $O(\frac{1}{\varepsilon}\cdot \frac{\Delta}{\gamma \cdot \ln \Delta})$. That is because by Lemma \ref{lem:parameterized ta few time steps}, objects' realizations come from at most $O(\log_{n^\varepsilon}n) = O(\frac{1}{\varepsilon})$ different time steps, and for each time step we are using $O(\frac{\Delta}{\gamma \cdot \ln \Delta})$ colors. 

    As for the total computation, we first observe that the number of adversarially sampled objects is no more than $2q$ since each edge insertion only resamples two vertices. Now define $n_t = |V(S_t)|$ - the number of nodes being sampled at time $t$. By Theorem \ref{thm:triangle-free coloring sublinear alg}, at time $t$, with high probability, the algorithm uses $\widetilde{O}(n_t^{3/2 + \gamma})$ computation, hence

    \begin{align*}
        \mathrm{Total\ Computation} &= \sum\limits_{t \in [T]}\widetilde{O}(n_t^{3/2 + \gamma})\\
        &\leq n^{1/2 + \gamma}\cdot \sum\limits_{t \in [T]}\widetilde{O}(n_t)
    \end{align*}
    where the last inequality is because $n_t^{1/2 + \gamma}\leq n^{1/2 + \gamma}$.
    
    Furthermore, Theorem \ref{thm:parameterized temporal aggregation} gives us that $\sum\limits_{t \in [T]}n_t = O(q\cdot n^{\varepsilon} \cdot \log n)$. Thus we obtain
    \begin{align*}
        \mathrm{Total\ Computation} &= n^{1/2 + \gamma}\cdot \widetilde{O}(q\cdot n^{\varepsilon}),
    \end{align*}
    which completes the proof.
\end{proof}

\subsection{Rectifying Simple Dynamic Spanners}
\label{sec:appendix:proofs:rectifying}
\chargingJobs*
\begin{proof}
For $u\in U$, define \begin{align*}r_1(u) &= \begin{cases}
    |U_i|, \text{ s.t. $u$ appears in $S_i$ for the first time}\\
    +\infty, \text{ if $u$ does not appear in any $S_i$}
\end{cases}\\ \text{and similarly,}&\\
r_2(u) &= \begin{cases}
    |U_i|, \text{ s.t. $u$ appears in $S_i$ for the second time}\\
    +\infty, \text{ if $u$ does not appear two times}
\end{cases}\\
\text{and also}&\\
r(u) &= \min(r_1(u), r_2(u))
\end{align*}

Then 
\begin{align*}
    \sum\limits_{i=1}^n\frac{|S_i|}{|U_i|} = \sum\limits_{i = 1}^m\sum\limits_{u \in S_i}\frac{1}{|U_i|} = \sum\limits_{u \in U}\frac{1}{r_1(u)} + \frac{1}{r_2(u)}\leq 2\sum\limits_{u \in U}\frac{1}{r(u)}
\end{align*}

We claim that for a given $k$ there are at most $k$ different $u$-s such that $r(u)\leq k$. To observe that, assume the contrary, namely, that there are $k + 1$ $u$-s, denote those $u_1, \ldots, u_{k+1}$, such that for all $i \in [k+1]$, $r(u_i) \leq k$. Define $f(u_i)$ to output the set in $\mathcal{U}$ of size $\leq k$ to which $u_i$ belongs and consider $u^\ast \in \{u_1, \ldots, u_{k+1}\}$ such that $U^\ast = f(u^\ast)$ has the maximal cardinality. Since all $U$-s in $\mathcal{U}$ can be ordered by inclusion, $U^\ast$ must be a superset for every $f(u_i)$. But $f(u_i)$ contains $u_i$, so $U^\ast$ must contain all $u_1, \ldots, u_{k+1}$ and hence be of cardinality at least $k + 1$, which is a contradiction. 

Therefore, 
\begin{align*}
    \sum\limits_{i=1}^n\frac{|S_i|}{|U_i|} \leq 2\sum\limits_{u \in U}\frac{1}{r(u)} \leq2\sum\limits_{k \in [n]}\frac{1}{k} = O(\log n).
\end{align*}

\end{proof}

\end{document}